%% file: main.tex
\documentclass[a4paper, amsfonts, amssymb, amsmath, reprint, showkeys, nofootinbib, twoside, superscriptaddress, prx]{revtex4-1}
\input{preamble}

\begin{document}
\title{Qubit Loss Inference with Stabilizer Codes without Leakage Detection Units
}
\author{Shin Nishio}
    \email[Correspondence email address: ]{parton@sfc.wide.ad.jp}
    \affiliation{University College London, Gower St, London WC1E 6BT, United Kingdom}
    \affiliation{Keio University, Yokohama, Kanagawa 223-8522, Japan}
\author{Takeaki Uno}
    \affiliation{National Institute of Informatics, 2-1-2 Hitotsubashi, Chiyoda-ku, Tokyo 101-8430, Japan}
\author{Fumiyoshi Kobayashi}
    \affiliation{Mercari, Inc., R4D, Tokyo 106-6118, Japan}
\author{Takahiko Satoh}
    \affiliation{Keio University, Yokohama, Kanagawa 223-8522, Japan}
\author{Dan E. Browne}
    \affiliation{University College London, Gower St, London WC1E 6BT, United Kingdom}

\begin{abstract}
Qubit loss occurs when the physical carrier of a qubit leaves the computational system without directly revealing the event's location. Such errors are a major obstacle to fault-tolerant quantum computation on platforms including photonic, neutral-atom, and trapped-ion systems. Loss locations are commonly identified using additional hardware operations such as leakage-detection units (LDUs), which introduce space-time overhead and may themselves become a source of error.

We investigate whether qubit loss on stabilizer codes can instead be inferred from syndrome data obtained through standard repeated stabilizer measurements. Under a non-entangling model for gates involving a lost qubit, we derive a sufficient condition for loss detectability in general stabilizer codes. The condition is based on the emergence of anticommutation between stabilizer checks after their support on the lost qubits is removed. 
By using that condition, 
we formulate the exact loss-inference problem using the observed set of non-deterministic checks together with its maximum-likelihood formulation. We then relax the problem to the minimum set cover problem with a greedy heuristic algorithm.
We evaluate the resulting inference and loss-correction protocols on the rotated surface code via circuit-level noise simulations for trapped-ion and neutral-atom platforms. On both platforms, inference-based and adaptive protocols reduce the logical error rate relative to a noisy-LDU baseline in the low-to-moderate loss-rate regime relevant to near-term hardware, while requiring fewer space-time overheads.
\end{abstract}

\keywords{Qubit Loss Errors, Quantum Error Correcting Code, Loss Inference, Leakage Detection Unit}

\maketitle
\section{Introduction}

Quantum systems are subject to various kinds of noise, and \textit{loss and leakage errors} can be particularly damaging~\cite{google2025quantum, bluvstein2026fault}. A leakage error takes the physical qubit out of the qubit subspace, while a loss error is a special kind of leakage that removes the qubit altogether. 
\emph{Loss errors} do not merely affect the codewords. They also effectively reduce the effective code distance~\cite{gu2025optimizing, chang2025surface}, making long-time quantum computation impossible without an enormous amount of costly reloading.
We focus on qubit loss, a form of leakage in which the physical carrier becomes unavailable. When the location of such an event is known, it can be treated as an erasure. Here we ask how much location information can instead be inferred directly from syndrome data without dedicated detection hardware.

In general, when a leakage error occurs, one converts it into an erasure using a quantum circuit called a \emph{leakage detection unit} (LDU)~\cite{gottesman1997stabilizer, chow2024circuit, perrin2025quantum}. However, because an LDU circuit needs to be executed in addition to the standard syndrome measurements, it can lead to a longer code cycle, demand more auxiliary qubits, and become a new source of error.

In this work, we formalize the condition for detectable loss errors using the statistics of multiple rounds of stabilizer measurements alone, and we formulate the inference problem. 
Fig.~\ref{fig:loss_flow} shows a flow of loss correction with the leakage detection unit or the loss inference.
Table~\ref{tab:ancilla_comparison} summarizes the auxiliary-qubit overhead required by representative approaches. Conventional circuit-LDU requires $n$ additional auxiliary qubits dedicated to leakage detection, while our inference-based approach eliminates these auxiliaries entirely at the expense of additional reloading caused by wrong inference.
 
Our paper is organized as follows: Sec.~\ref{sec:preliminary} gives a brief overview of the loss error and the standard hardware-based countermeasure to it, LDU. Then Sec.~\ref{sec:assumption} gives an assumption on the behavior of the unitary gates on the lost qubit, where CX gates on a lost qubit end up with a no-entangling gate followed by a Pauli channel on a surviving qubit. This assumption is natural in some physical systems.
Based on that assumption, we give the condition on the detectable loss errors for arbitrary stabilizer codes in Sec.~\ref{sec:condition}. Sec.~\ref{sec:problem} defines the maximum-likelihood loss inference problem for general stabilizer codes. We then give a relaxed problem as the minimum set cover problem in Sec.~\ref{sec:min_set_cov}. The set cover problem has a greedy-based heuristic algorithm with a bounded approximation ratio, which we describe in Sec.~\ref{sec:heuristic}. In Sec.~\ref{sec:eval}, we evaluate the loss correction protocol on rotated surface codes with loss inference from the perspective of the accuracy of the inference and also the logical error rate. 
We finally summarize and discuss the practical flow for a loss-tolerant system in Sec.~\ref{sec:conclusion}.

\begin{figure*}[htbp]
    \centering
    \includegraphics[width=0.9\linewidth]{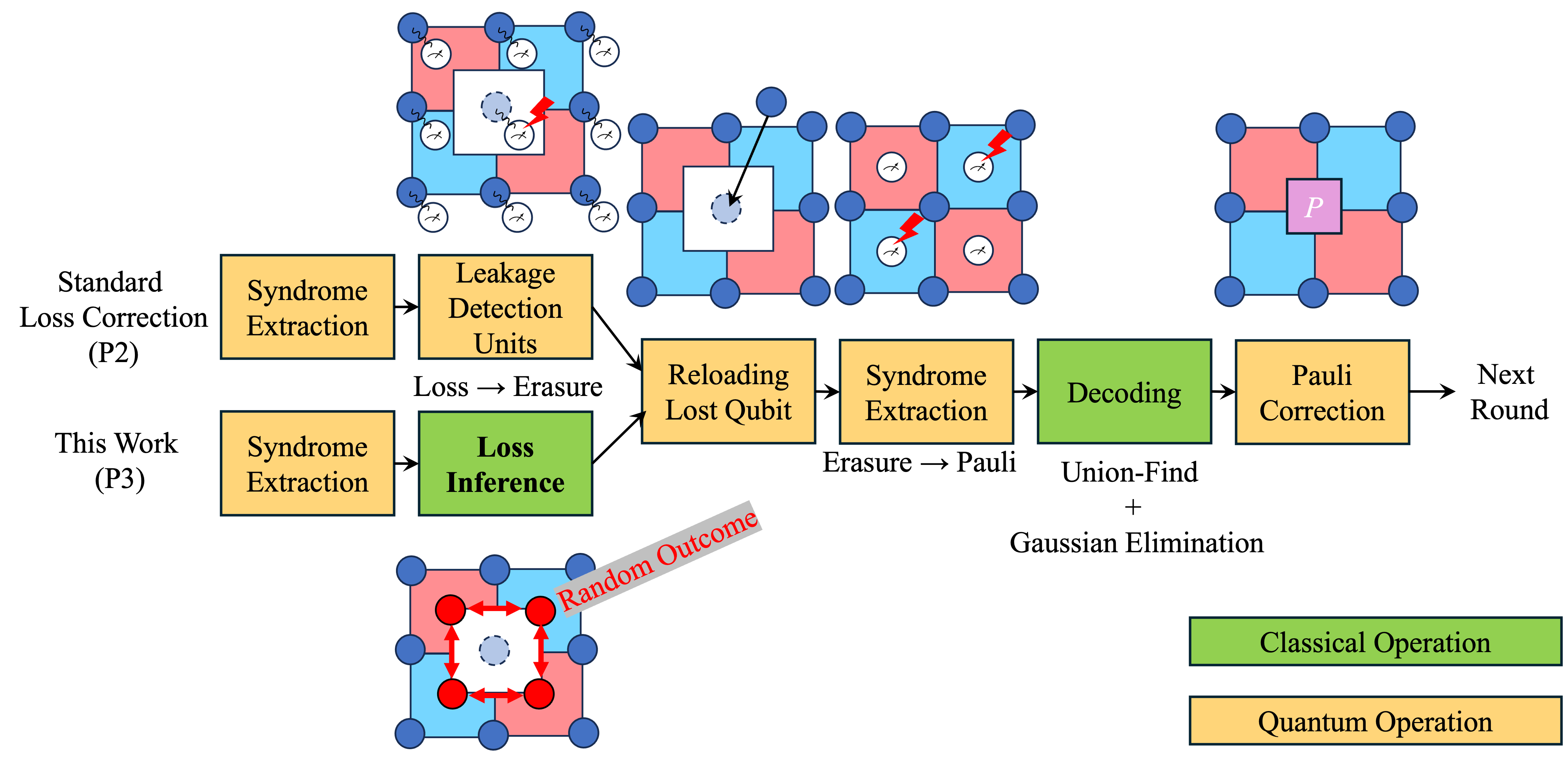}
    \caption{Loss correction protocol flows. 
    The standard protocol to deal with loss errors is converting the loss error into an erasure at a known location via leakage detection units.
    Instead of leakage detection units, we propose a protocol with loss inference, which no longer requires running an additional circuit, which reduces overhead and eliminates a source of noise. After applying one of them, one can apply reloading, syndrome extraction, decoding, and Pauli error correction, or Pauli frame update. Note that surface code is used as an example, but both protocols can be easily generalized to arbitrary stabilizer codes by replacing the decoder.}
    \label{fig:loss_flow}
\end{figure*}

\begin{table}[htbp]
\centering
\caption{Number of auxiliary qubits required by each leakage detection method per round for a $[\![n,k,d]\!]$ code with $l$ loss errors. FP and FN stand for the number of false positive and false negative results from the loss inference algorithm.}
\begin{tabular}{lccc|c}
\hline
 & Stab. Check & LDU & Reloading & Sum \\
\cline{1-5}
Circuit LDU
& $n-k$ & $n$ & $l$ & $2n-k+l$ \\
SWAP LDU
& $n-k$ & $n$ & -- & $2n-k$ \\
Loss Inference
& $n-k$ & -- & $l+\mathrm{FP}-\mathrm{FN}$
& 
\begin{tabular}{c}
$n-k+l$\\
$+\mathrm{FP}-\mathrm{FN}$ \\
\end{tabular}\\ \hline\hline
\end{tabular}
\label{tab:ancilla_comparison}
\end{table}

\section{Preliminaries}
\label{sec:preliminary}
In this work, we deal with loss errors on data qubits in stabilizer codes. 
Before discussing how such errors may be detected, we first clarify related terms that are sometimes used interchangeably in the literature.

\subsection{Leakage Error and Terminologies}
\label{subsec:terminology}
A \emph{leakage error} occurs when the state of a physical qubit leaves the designated computational subspace and populates states outside that subspace~\cite{wood2018quantification}. Equivalently, writing the full physical Hilbert space as \(\mathcal{H}_{\mathrm{comp}}\oplus \mathcal{H}_{\mathrm{leak}}\), a leakage channel is any channel that transfers population from the former to the latter. Leakage is therefore more general than ordinary in-subspace noise: once leakage occurs, subsequent gates and measurements need not act according to the intended qubit model. Such errors arise in several platforms, including superconducting qubits~\cite{wood2018quantification}, spontaneous scattering or decay in trapped ions~\cite{kang2023quantum}, and decay from metastable qubit or transitions induced by blackbody radiation of Rydberg levels on neutral atom systems~\cite{wu2022erasure, scholl2023erasure}.

A \emph{loss error} is a more specific situation in which the physical carrier of the qubit becomes unavailable to subsequent operations.
Examples include atom loss from an optical tweezer~\cite{vala2005quantum, kobayashi2025erasure} or photon loss in photonic encodings~\cite{wasilewski2007protecting}. 
In this sense, loss can be regarded as a particular form of leakage that removes the qubit from the effective computational register altogether. 
Throughout this paper, our focus is on such loss errors on data qubits. 

An \emph{erasure error} is an error whose location is known (heralded)~\cite{grassl1997codes}. 
More precisely, in the quantum-erasure model, one allows an arbitrary error on a known qubit location, which makes erasures substantially easier to handle than errors at unknown locations. 
This distinction is fundamental in quantum coding theory: for example, the quantum erasure channel admits shorter codes than the corresponding unknown-location error model~\cite{grassl1997codes}. 
In many fault-tolerant architectures, a detected loss or detected leakage event is therefore treated as an erasure. 
A closely related static setting arises in superconductivity architectures with known hardware \emph{fabrication errors}, also known as \emph{dropouts}, such as missing or unusable qubits and couplers due to imperfect device yield, where the locations of the defects are known and the code or syndrome-extraction circuit can be adapted accordingly~\cite{auger2017fault, debroy2025luci}.
There have been efforts to convert various types of unknown position errors into erasure errors physically~\cite{wu2022erasure, kang2023quantum}, and it is known that this can reduce the computational complexity of the decoding problem~\cite{delfosse2020linear}, hence relaxing the required responsiveness of online decoders.

A \emph{deletion error} refers to the removal of a subsystem together with the loss of its positional alignment, so that the indices of the remaining subsystems become ambiguous~\cite{leahy2019quantum, nakayama2020single, hagiwara2025introduction}. This is conceptually distinct from qubit loss, which concerns the physical absence of a carrier without necessarily invoking such synchronization ambiguity. We mention deletion only to avoid this terminological confusion, since our focus is on qubit loss.

With this terminology, the central question of this paper can be stated succinctly: when loss occurs without dedicated hardware that explicitly flags its location, under what conditions can stabilizer measurements still reveal where the loss has happened? 
The following subsections review the standard hardware-based approach to this question and explain why it is useful to study alternatives that do not rely on it.

\subsection{Standard Hardware-based Approaches}
\label{subsec:ldu}
The standard hardware-based approach to loss and leakage is to convert them into erasures by explicitly identifying the affected qubit location during fault-tolerant operation. In practice, this is achieved either through direct measurement primitives that distinguish computational states from loss, or through auxiliary circuit gadgets that detect or remove leakage before subsequent syndrome-extraction rounds.

One representative strategy is to use measurement operations with a richer outcome alphabet, for example, measurements that distinguish \(\{0,1,\mathrm{Loss}\}\) rather than only the computational basis outcomes. An example of this measurement is the three-state readout in a superconducting qubit~\cite{chen2016measuring} and state-selective readout (SSR)~\cite{PhysRevA.108.03240, PRXQuantum.4.030337} for the neutral atom qubit as discussed in the Appendix~\ref{app:atom-SSR}. When available, such measurements directly reveal the location of the loss event and allow the decoder to treat it as an erasure~\cite{stricker2020experimental, de2025mid, baranes2026leveraging}. 
Varbanov et al. inferred the time and location of leakage in a transmon-based surface code from neighboring stabilizer syndrome and analog three-state readout~\cite{varbanov2020leakage}. 
SSR is particularly useful for detecting auxiliary-qubit loss and is complementary to our proposed method for inferring data-qubit loss.

A complementary strategy is to use dedicated leakage-handling gadgets. A leakage-reduction unit (LRU) maps an input suffering possible leakage back into the computational subspace, at the cost of at most an ordinary in-subspace fault; in the framework of Aliferis and Terhal, quantum teleportation provides a natural universal implementation of such an LRU~\cite{aliferis2005fault}.
Syndrome information has also been exploited to identify leakage. Recent approaches infer or predict leakage locations from characteristic syndrome patterns to reduce the required number of LRUs~\cite{vittal2023eraser, mude2025accurate}. These works employ a phenomenological leakage model in which leakage either propagates or induces effective Pauli errors during two-qubit gates.

More directly relevant to loss detection are leakage-detection units (LDUs)~\cite{gottesman1997stabilizer, chow2024circuit, perrin2025quantum}, which aim to flag the affected location so that the event can be handled as an erasure in the subsequent decoding stage. Variants such as SWAP-LDU further combine detection with replacement or reloading of the affected subsystem~\cite{chow2024circuit, liu2026achieving}.

\subsection{Limitations and Costs of LDUs}
\label{subsec:cost}
Despite their advantages, LDUs and related hardware-assisted schemes come with nontrivial costs. First, they generally increase circuit depth or code-cycle duration, since dedicated detection, replacement, or reinitialization operations must be inserted into the standard syndrome-extraction schedule. Second, they may require additional auxiliary qubits, measurement capabilities, or reload primitives that are not part of the baseline architecture. This overhead is not merely static: in architectures with qubit replacement or atom reloading, a larger auxiliary footprint also increases the number of subsystems that must be provisioned, monitored, and, when necessary, replenished during repeated syndrome-extraction cycles. In neutral-atom platforms, this burden can be especially significant because atom replacement has historically been a practical bottleneck for deep circuits, and avoiding full-array reloads is itself a major architectural concern~\cite{li2025fast,litteken2022reducing}.
Third, the added operations themselves introduce new opportunities for faults, so the mechanism intended to mitigate leakage or loss can become an additional source of error.

These costs motivate the present perspective. Rather than assuming a dedicated leakage-detection mechanism, we ask to what extent loss locations can be inferred directly from the statistics of repeated stabilizer measurements. This is the sense in which our approach is LDU-free. To study syndrome-based inference, we next introduce the gate model assumed throughout this work.

While this manuscript was in preparation, Wang et al. independently proposed a spatiotemporal graph-neural-network decoder that infers qubit loss from repeated stabilizer measurements designed for surface codes~\cite{wang2026ai}. In contrast, our work derives an algebraic sufficient condition for loss detectability and formulates the resulting inference problem for general stabilizer codes.

\section{Assumption of the System}
\label{sec:assumption}
In this work, we consider loss errors on the data qubits of a stabilizer code. 
The action of a unitary gate on a qubit that has undergone a loss error is, in general, no longer well defined on the original computational Hilbert space.

We assume the following behavior of gates acting on lost qubits as it behaves in some physical systems, which will be used later in Sec.~\ref{sec:condition}.
After a loss event, we do not track the state of the lost qubit. Instead, we model the effect on the surviving qubit as follows.

\begin{assumption}[Non-entangling assumption]
Suppose the initial joint state is separable and given by $\ket{\psi}_L \otimes \ket{\phi}_S,$
where qubit $L$ has been lost and qubit $S$ survives without error.
When a CX gate is applied between $S$ and $L$, the resulting operation is assumed to be non-entangling. The effect on the surviving qubit $S$ is modeled as a Pauli channel,
\[
\mathcal{E}_S(\rho)=\sum_{P\in\{I,X,Y,Z\}} p_P\, P\rho P,
\qquad \sum_P p_P=1.
\]
\label{assumption_noise}
\end{assumption}

Under Assumption~\ref{assumption_noise}, the failed syndrome extraction circuit for generator $s_i$ implements a measurement of a punctured operator $s'_i$, up to a sign and/or a local Pauli frame on the surviving subsystem.

For example, Fig.~\ref{fig:extraction_circuit}(a) shows the syndrome extraction circuit for $Z_1Z_2Z_3Z_4$ stabilizer operator without any error and Fig.~\ref{fig:extraction_circuit}(b) shows the circuit for the case having a loss error on the third qubit. It failed to measure the original operator and is now measuring the effective punctured operator which can be regarded as a gauge operator~\cite{poulin2005stabilizer}, which is not necessary in the stabilizer group. In this case, the operator is $\pm Z_1 Z_2 Z_4$ up to the sign, which is determined by the Pauli error. 

We have shown that this assumption is satisfied for the ideal loss models of the M\o lmer--S\o rensen gate on trapped ions in Appendix~\ref{app:MS} and Rydberg-blockade-based CZ gate on neutral atoms in Appendix~\ref{app:blockade-CZ}.
We also remark that this effective no-entangling assumption has been widely discussed in recent works on loss correction on photonic system~\cite{ralph2015photon}, neutral atoms~\cite{perrin2025quantum, liu2026achieving}(no Pauli channel following the loss), and trapped ions~\cite{brown2019leakage} (Pauli $Z(-\pi / 2)$ / $X(-\pi / 2$) error following the CX gate occurs on the survived qubit depends on whether it is control/target qubit).

\begin{figure}[htbp]
\subfloat[]{%
  \includegraphics[clip,width=0.47\columnwidth]{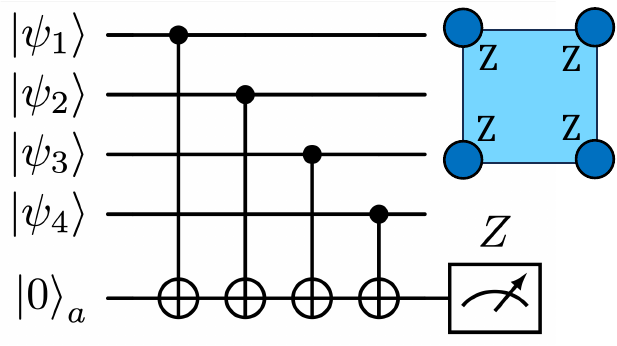}
}
\subfloat[]{%
  \includegraphics[clip,width=0.52\columnwidth]{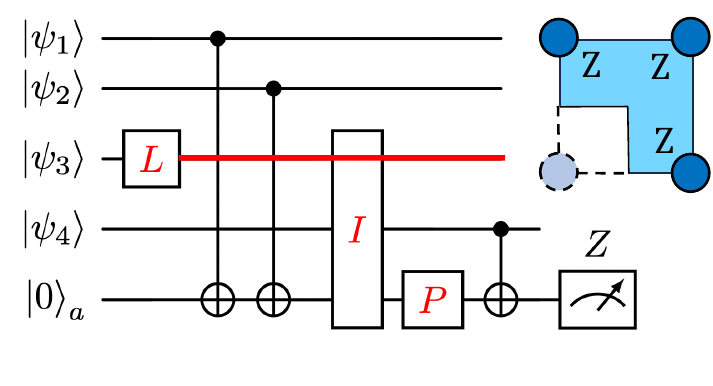}
}
\caption{
(a)  A quantum circuit for stabilizer measurement on $Z_1 Z_2 Z_3 Z_4$ operator.
(b) A quantum circuit intended to execute stabilizer measurement for $Z_1 Z_2 Z_3 Z_4$ operators on a punctured system, which resulted in measuring an effective gauge operator $\pm Z_1 Z_2 Z_4$ where the sign is determined by the Pauli error operator $P$. }
\label{fig:extraction_circuit}
\end{figure}

In addition to the gate model, we also specify the temporal model of loss events during repeated syndrome extraction. Throughout this work, we assume that loss events occur between syndrome-extraction rounds with a fixed CX gate order. Consequently, the set of lost qubits remains fixed throughout each syndrome-extraction circuit, and each stabilizer measurement acts on a fixed punctured code during that round.

Assumption~\ref{assumption_noise} also modifies the propagation of Pauli faults on auxiliary qubits through the syndrome-extraction circuit. Under this assumption, interactions involving lost qubits become non-entangling, so hook-error propagation~\cite{tomita2014low, 7lm4-3bnh, hirai2026hookssurfacecodedistancepreserving} is determined by the remaining sequence of CX interactions rather than the original circuit. As a consequence, the support and, for topological codes such as the rotated surface code, the orientation of hook errors may differ from those in the loss-free circuit. Since this effect is not required for the loss-inference framework developed in the main text, we defer a detailed discussion to Appendix~\ref{app:hook_error}.

\section{Sufficient Condition on Detectable Loss Errors}
\label{sec:condition}

In this section, we derive a commutation-based criterion for the detectability of loss errors under the noise model introduced in Sec.~\ref{sec:assumption}.
The key idea is that, after puncturing the lost qubits, the commutation relations among the effective check operators may change.
If anticommutation emerges among the post-loss check operators, then the corresponding syndrome outcomes can no longer remain deterministic, which makes the loss pattern detectable in principle.

For an $n$-qubit Pauli operator $s$, let $(s)_q \in \{I,X,Y,Z\}$ denote its local Pauli component on qubit $q$.
For a set of lost qubits $L \subseteq \{1,\dots,n\}$, let $s'$ denote the punctured operator obtained from $s$ by removing its local Pauli factors on all qubits in $L$.

\begin{theorem}[Parity Criterion for Anticommutation of Punctured Checks]
\label{thm:parity_criterion}
Let $\mathscr{S}=\langle s_1,\dots,s_{n-k}\rangle$ be a stabilizer group, and let $L$ be the set of lost qubits.
For two stabilizer generators $s_i$ and $s_j$ with $i\neq j$, let $s_i'$ and $s_j'$ denote the corresponding punctured operators.

Define
$$ A_{ij}:=\{q\in \{1,\dots,n\} \mid (s_i)_q (s_j)_q = - (s_j)_q (s_i)_q \},$$
namely, the set of qubits on which the single-qubit Pauli factors of $s_i$ and $s_j$ anticommute.
Then the punctured operators $s_i'$ and $s_j'$ anticommute if and only if $|L \cap A_{ij}|$
is odd.
\end{theorem}

\begin{proof}
Since $s_i$ and $s_j$ are stabilizer generators, they commute.
For Pauli operators, this is equivalent to the statement that the number of qubits on which $(s_i)_q$ and $(s_j)_q$ anticommute locally is even.
Thus, $|A_{ij}| \equiv 0 \pmod 2.$

After puncturing the lost qubits in $L$, the operators $s_i'$ and $s_j'$ anticommute locally exactly on the qubits in $A_{ij}\setminus L$.
Therefore, $s_i'$ and $s_j'$ anticommute if and only if
$|A_{ij}\setminus L|$ is odd.

Since $|A_{ij}\setminus L| = |A_{ij}| - |A_{ij}\cap L|$, we obtain modulo $2$ that
\[
|A_{ij}\setminus L|
\equiv
|A_{ij}\cap L|
\pmod 2,
\]
because $|A_{ij}|$ is even.
Hence $s_i'$ and $s_j'$ anticommute if and only if $|A_{ij}\cap L|$ is odd.
\end{proof}

Theorem~\ref{thm:parity_criterion} gives an algebraic characterization of when anticommutation emerges among the post-loss check operators.
We next show that such anticommutation necessarily leads to non-deterministic syndrome outcomes.

\begin{lemma}[Anticommuting Observables Cannot Be Simultaneously Deterministic]
\label{lem:anticommuting_nondeterministic}
Let $A$ and $B$ be Hermitian Pauli operators satisfying $AB=-BA$.
Then no quantum state can be a simultaneous eigenstate of both $A$ and $B$.
In particular, there is no state for which the outcomes of both observables are simultaneously deterministic.
\end{lemma}

\begin{proof}
Suppose, for contradiction, that there exists a state $\ket{\psi}$ such that
$A\ket{\psi}=a\ket{\psi}, B\ket{\psi}=b\ket{\psi}$
for some $a,b\in\{\pm1\}$.
Then
\[
AB\ket{\psi} = ab\ket{\psi}.
\]
On the other hand, since $AB=-BA$,
\[
AB\ket{\psi}
=
-BA\ket{\psi}
=
-ab\ket{\psi}.
\]
Hence $ab=-ab$,
which is impossible.
Therefore, no such state exists.
\end{proof}

Lemma~\ref{lem:anticommuting_nondeterministic} implies that once a pair of punctured check operators anticommutes, the corresponding syndrome outcomes cannot both remain deterministic. This is closely related to the well-known fact that eigenbases of anticommuting Pauli operators are mutually unbiased, implying that an eigenstate of one operator yields uniformly random outcomes when measured in the eigenbasis of the other~\cite{lawrence2004mutually}.
This immediately leads to the following detectability statement for code states.

\begin{theorem}[Anticommutation Implies Detectability of Loss]
\label{thm:detectability_code_state}
Assume that, before the loss event, the system is in a code space stabilized by the original stabilizer group $\mathscr{S}$.
Let $L$ be a loss pattern, and suppose that the corresponding punctured check operators contain a pair $s_i',s_j'$ such that
\[
s_i's_j'=-s_j's_i'.
\]
Then the post-loss syndrome outcomes cannot all remain deterministic in the same way as in the loss-free case.
Consequently, the syndrome statistics after the loss event differ from those of the loss-free case, and the loss pattern $L$ is therefore detectable in principle by repeated syndrome extraction.
\end{theorem}

\begin{proof}
Before the loss event, the system is assumed to be stabilized by $\mathscr{S}$.
Hence, each stabilizer generator has a deterministic measurement outcome $+1$ in the loss-free case.

Now suppose that, after puncturing the lost qubits, there exists an anticommutative pair $s_i',s_j'$.

By Lemma~\ref{lem:anticommuting_nondeterministic}, there is no quantum state for which both $s_i'$ and $s_j'$ have deterministic outcomes simultaneously.
Therefore, the post-loss state cannot reproduce the same deterministic syndrome behavior as the original loss-free code state.

It follows that, under repeated syndrome extraction, the corresponding checks exhibit random outcomes and therefore the syndrome statistics after loss must differ from those of the loss-free case. Therefore, the loss pattern is detectable in principle by repeated syndrome extraction.
\end{proof}

Combining Theorems~\ref{thm:parity_criterion} and \ref{thm:detectability_code_state}, we obtain the following sufficient criterion for loss detectability.

\begin{corollary}[Parity-Based Sufficient Condition for Detectable Loss]
\label{cor:parity_detectable}
Assume that, before the loss event, the system is in a code state stabilized by $\mathscr{S}$.
If there exists a pair of stabilizer generators $s_i,s_j$ such that
\[
|L\cap A_{ij}|
\]
is odd, then the loss pattern $L$ is detectable in principle by repeated syndrome extraction.
\end{corollary}

\begin{proof}
By Theorem~\ref{thm:parity_criterion}, the condition that $|L\cap A_{ij}|$ is odd implies that the punctured operators $s_i'$ and $s_j'$ anticommute.
The conclusion then follows immediately from Theorem~\ref{thm:detectability_code_state}.
\end{proof}

For example, Fig.~\ref{fig:commutation} shows a commutation relationship for two stabilizer checks $ZZZZ$ and $XXXX$ when a loss error occurred on data qubits.

\begin{figure}
    \centering
    \includegraphics[width=0.7\linewidth]{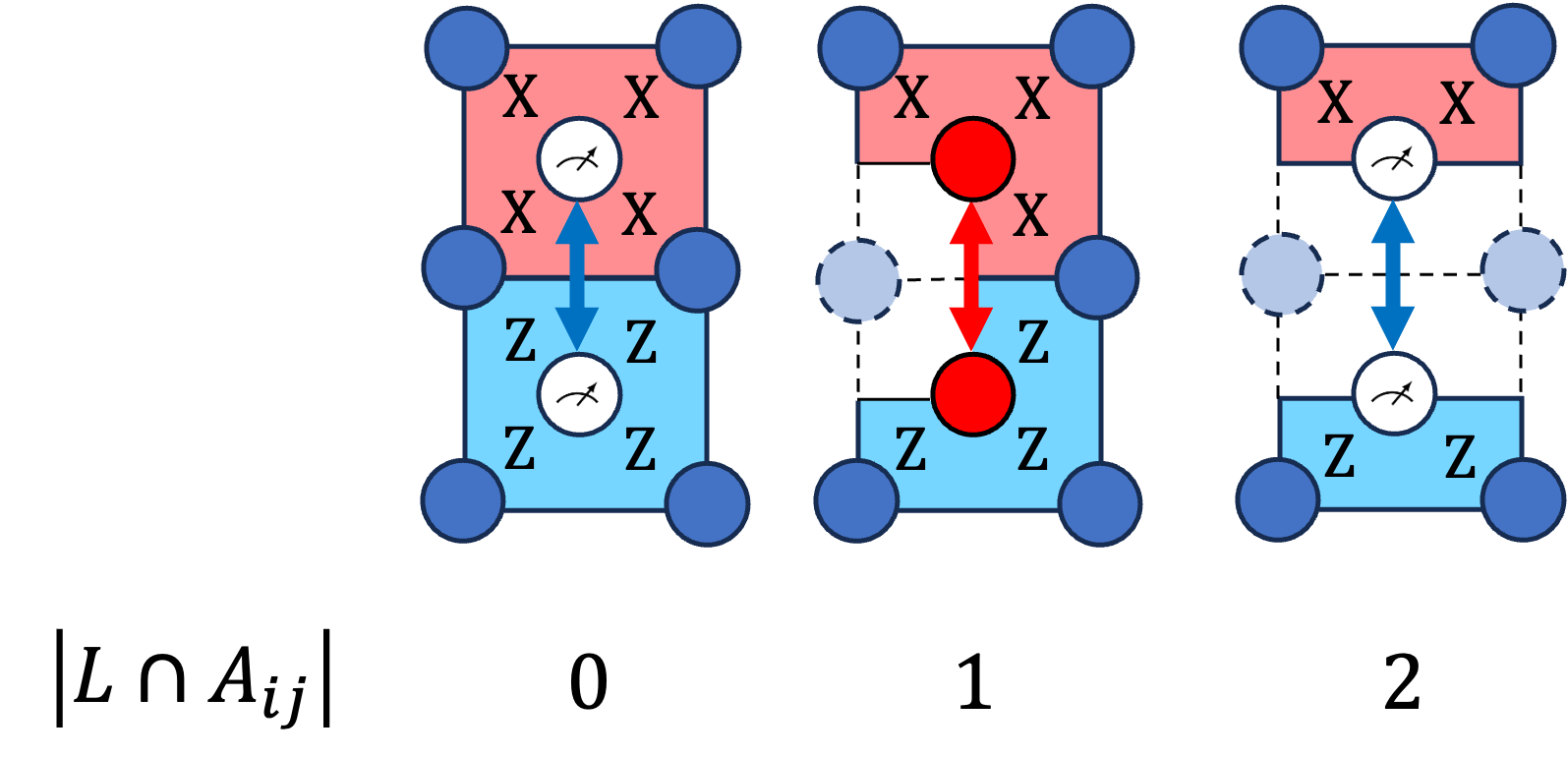}
    \caption{The commutation (blue double arrow) / anticommutation (red double arrow) relations of a pair of check operators $s_i$ and $s_j$. Blue vertices represent surviving data qubits and the dotted light blue vertices represent the lost qubits. Each vertex on the face is the auxiliary qubit for a check operator.
    (Left) Precondition: $A_{ij} = 2 \equiv 0 \pmod 2 $ so that two checks $s_i$ and $s_j$ commute. (Center) For the case of single loss, $|A_{ij}\cap L| \equiv 1 \pmod 2$, and so two checks anticommute. This loss is detectable. (Right) For an even number of losses on $A_{ij}$, two checks commute, therefore not detectable from the statistics of two stabilizer checks.}
    \label{fig:commutation}
\end{figure}

Corollary ~\ref{cor:parity_detectable} gives a sufficient condition for detecting loss from repeated stabilizer measurements. The syndrome-record distribution differs from that of the loss-free case. It does not imply unique identification of the loss pattern or recovery of the encoded state. Appendix \ref{app:operational-kl} formalizes this distinction in relation to the Knill–Laflamme framework~\cite{knill1997theory}.

In this sense, the loss event does not merely modify the measured stabilizer operators but can also be viewed as inducing a new gauge structure.
That is, the newly generated post-loss check operators can be regarded as gauge operators, so that the syndrome-extraction process effectively defines a dynamical code~\cite{Fu2025errorcorrectionin} whose operator structure changes after the loss event.\\

Corollary~\ref{cor:parity_detectable} addresses detectability, namely whether a loss event can be distinguished from the loss-free case. A stronger question is whether distinct loss events can be distinguished from one another using the resulting pattern of non-deterministic checks. We defer this inference question to Sec.~\ref{sec:problem}.

\section{the qubit loss inference problem}
\label{sec:problem}
The results of Sec~\ref{sec:condition} identify the key theoretical mechanism behind detectable qubit loss: puncturing by loss can change the commutation structure of stabilizer generators, and when this produces anticommutation among the effective post-loss checks, repeated syndrome measurements necessarily become random. Thus, the detectable-loss condition of Corollary~\ref{cor:parity_detectable} converts an algebraic property of punctured stabilizers into an experimentally accessible signature.

This observation motivates the qubit loss inference problem studied in this section. Rather than detecting loss by dedicated hardware, we ask whether the loss pattern can be inferred directly from repeated stabilizer measurements as an instance on a rotated surface code patch shown in Fig.~\ref{fig:problem_pic}. Concretely, the goal is to infer the set of lost data qubits from the subset of stabilizer checks whose measurement outcomes exhibit randomness.

The inference task consists of two stages. First, from a finite number of repeated syndrome-extraction rounds, classifying each stabilizer check as deterministic or random. Second, given the observed set of random checks, reconstruct a loss pattern that is consistent with the detectable-loss condition. We formulate this inference problem for general stabilizer codes.

\begin{figure}[htbp]
    \centering
    \includegraphics[width=0.7\linewidth]{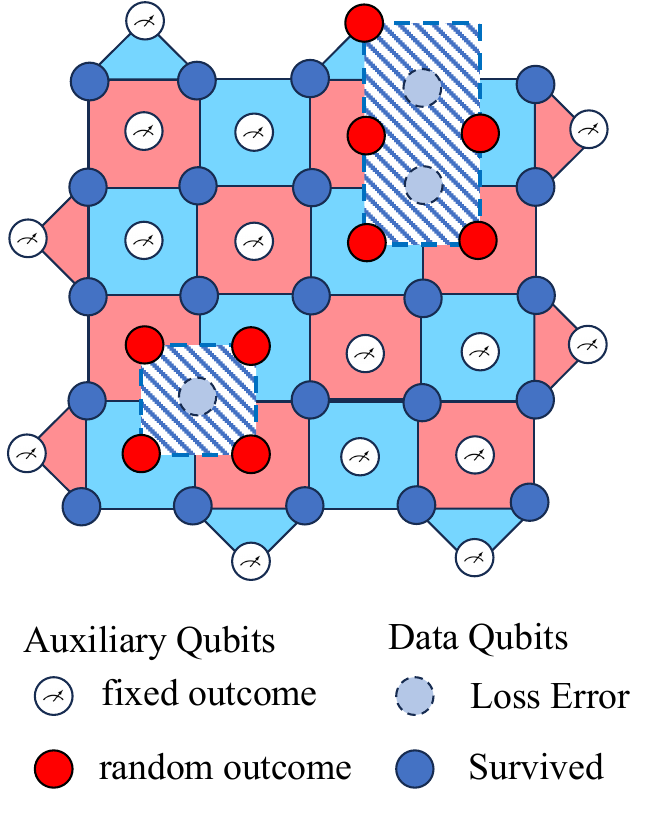}
    \caption{An instance of the loss inference problem on a rotated surface code. Given checks with random outcome depicted as red vertices, find the loss error pattern shown as the shaded area.}
    \label{fig:problem_pic}
\end{figure}

\subsection{Classification of Non-Deterministic Checks}
\label{subsec:classification}

Here we describe the classification method for random/deterministic checks
as it is used throughout the simulation study of Sec.~\ref{sec:eval}.

We collect $N$ repeated circuit shots. For shot $\alpha \in \{1,\dots,N\}$, check $i$, and round $r \in \{1,\dots,T\}$ within an observation window of $T$ rounds, let
$o_{i,r}^{(\alpha)} \in \{0,1\}$
denote the detector outcome of check $i$ at round $r$ in shot $\alpha$, i.e.\ the
indicator that the measured value of stabilizer check $s_i$ differs from its value at the
previous round. We first average over the shot ensemble to obtain an
empirical per-round flip rate,
\begin{equation}
  \hat p_{i,r} \;=\; \frac{1}{N}\sum_{\alpha=1}^{N} o_{i,r}^{(\alpha)},
  \label{eq:round-flip-rate}
\end{equation}
which estimates the true round-$r$ flip probability $\pi_{i,r}$ of check $i$
under the given scenario. Given a window width $W \le T$, we then take the
sliding average of $\hat p_{i,r}$ over the round axis and record its maximum
over all valid window placements $r_0 \in \{1,\dots,T-W+1\}$,
\begin{equation}
  q_i(W) \;=\; \max_{r_0} \; \frac{1}{W}\sum_{r=r_0}^{r_0+W-1} \hat p_{i,r}.
  \label{eq:windowed-score}
\end{equation}
Given a threshold $\tau$, check $i$ is classified as producing a random
outcome if
\begin{equation}
  q_i(W) \;\ge\; \tau,
  \label{eq:threshold-rule}
\end{equation}
and the observed random-check set is $R_{\mathrm{obs}} = \{\, i \mid q_i(W)
\ge \tau \,\}$. Unlike a symmetric randomness score, Eq.~\eqref{eq:threshold-rule} is a one-sided test on the flip rate itself rather than on its distance from $1/2$. This is sufficient to separate the two regimes relevant to Sec.~\ref{sec:assumption}. A loss-free check has
$\pi_{i,r}$ set by the background Pauli error rate and is $\ll \tau$, while a check made anticommuting by a loss event approaches $\pi_{i,r} \approx 1/2$ within the rounds in which the anticommutation is active. A fully deterministic check that flips on \emph{every} round ($q_i = 1$) would also be classified as random under Eq.~\eqref{eq:threshold-rule}, which a symmetric score would correctly reject; we do not observe this degenerate case arising from the noise model of Sec.~\ref{sec:assumption}, since puncturing only introduces anticommutation, not deterministic flipping, and we do not rely on the one-sided rule to reject it.

Extending the classifier of
Eq.~\eqref{eq:threshold-rule} to the genuinely online, single-shot setting is
left to future work, consistent with the outlook of
Sec.~\ref{sec:conclusion}.




\subsection{Loss Inference Problem for General Stabilizer Codes}
\label{sec:general_loss_inference}
Let $S=\langle s_1,\dots,s_{m}\rangle$ be a stabilizer code with $n$ data qubits and $m=n-k$ stabilizer generators.
After performing the classification in Sec.~\ref{subsec:classification}. Then, we obtain an observed random check set
$$R_{\mathrm{obs}} \subseteq \{1,\dots,m\}$$
where $i\in R_{\mathrm{obs}}$ means that the measurement outcome sequence of check $s_i$ is classified as random.

For a candidate loss pattern $L\subseteq \{1,\dots,n\}$, define the predicted random-check set by
\begin{eqnarray}
&&\mathcal{R}(L) :=\nonumber\\
&&\left\{ i \in \{1,\dots,m\}\;\middle|\; \exists j\neq i \text{ s.t. } |L\cap A_{ij}| \equiv 1 \pmod 2 \right\}\nonumber\\
&& 
\label{eq:predicted_random_checks}
\end{eqnarray}

where $A_{ij}$ is the set of qubits on which $s_i$ and $s_j$ anticommute locally, as defined in Theorem~\ref{thm:parity_criterion}.

This leads to the following loss inference problem.

\begin{tcolorbox}
\begin{problem}[Exact Loss Inference Problem]\label{prob:loss_inf}
Given an observed random-check set $R_{\mathrm{obs}}$, \\
find a loss pattern 
$L\subseteq \{1,\dots,n\}$ such that $\mathcal{R}(L)=R_{\mathrm{obs}}$.
\end{problem}
\end{tcolorbox}

Under an i.i.d.\ loss model with loss probability $p_{\mathrm{loss}}\in(0,1)$, the maximum-likelihood loss inference problem is
\begin{eqnarray}
L^\star\nonumber
&=&
\arg\max_{L:\,\mathcal{R}(L)=R_{\mathrm{obs}}}\Pr(L)\nonumber\\
&=&\arg\max_{L:\,\mathcal{R}(L)=R_{\mathrm{obs}}}
p_{\mathrm{loss}}^{|L|}(1-p_{\mathrm{loss}})^{n-|L|}.\nonumber
\label{eq:ml_loss_inference}
\end{eqnarray}
For the physically relevant regime $p_{\mathrm{loss}}<1/2$, this is equivalent to the minimum-weight problem
\begin{equation}
L^\star
=
\arg\min_{L:\,\mathcal{R}(L)=R_{\mathrm{obs}}}|L|.\nonumber
\label{eq:min_weight_loss_inference}
\end{equation}

In general, the map $L\mapsto \mathcal{R}(L)$ is nonlinear because the randomness of check $i$ is determined by the existence of at least one witness $j$ satisfying an odd-parity condition. 

\begin{remark}[Identifiability of loss errors]
In general, the map $L \mapsto R(L)$ is not necessarily injective. Therefore, distinct loss patterns may yield the same observed random-check set $R_{\mathrm{obs}}$, and the exact loss inference problem need not admit a unique solution. In particular, single-qubit loss events are uniquely identifiable from the observed random-check set if and only if the map $q \mapsto R(\{q\})$ is injective.
\end{remark}

It is therefore useful to introduce the \emph{check-interaction graph}
\[
G=(V,E),\qquad V=\{1,\dots,m\},
\]
where each vertex is the check and an edge $(i,j)\in E$ is present whenever $A_{ij}\neq\emptyset$.
Each data qubit $q\in\{1,\dots,n\}$ induces a subset of edges
\begin{equation}
E_q := \{(i,j)\in E \mid q\in A_{ij}\},\nonumber
\label{eq:qubit_edge_set}
\end{equation}
meaning that losing qubit $q$ flips the parity on every edge in $E_q$. For example, Fig.~\ref{fig:rotation_for_interaction} is a rotated surface code~\cite{bombin2007optimal}. Fig.~\ref{fig:interaction} shows its check interaction graph and $E_5$. This check-interaction graph can be obtained by Algorithm~\ref{alg:make_interaction}.

\begin{figure}[htbp]
    \centering
    \includegraphics[width=0.4\linewidth]{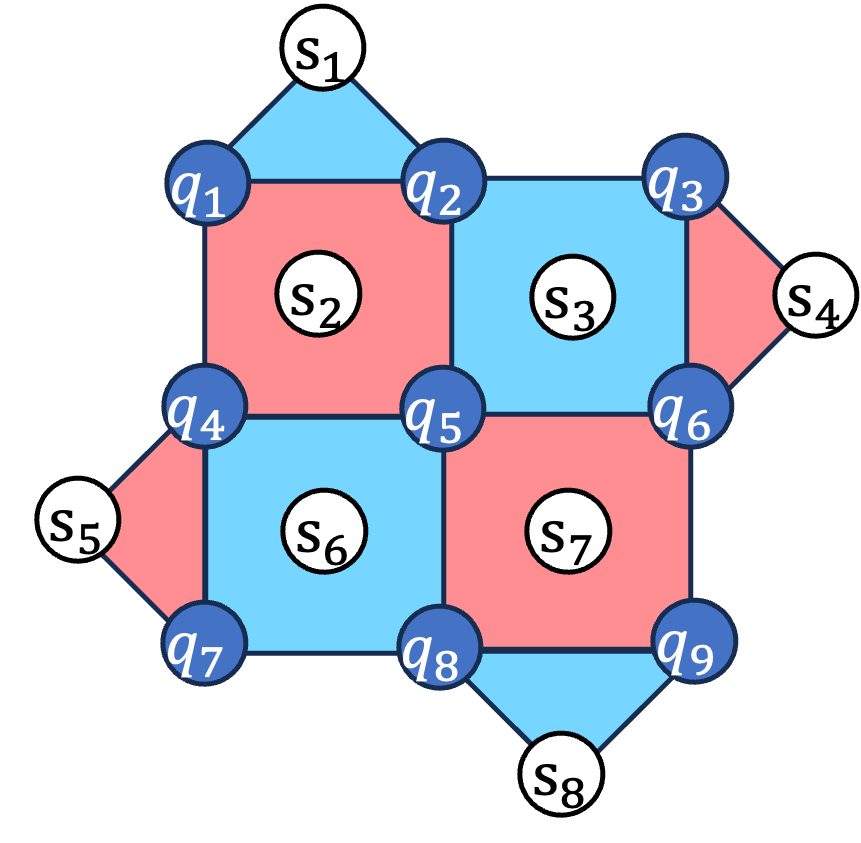}
    \caption{A rotated surface code.}
    \label{fig:rotation_for_interaction}
\end{figure}
\begin{figure}[htbp]
\centering
    \subfloat[]{%
  \includegraphics[clip,width=0.7\columnwidth]{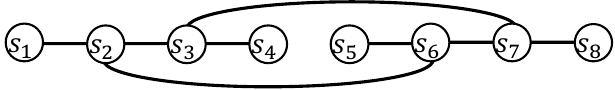}
    }\\
    \subfloat[]{%
  \includegraphics[clip,width=0.7\columnwidth]{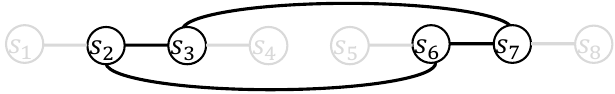}
    }
\caption{(a) The check interaction graph of the rotated surface code shown in Fig.~\ref{fig:rotation_for_interaction}. (b) $E_5$, the induced subset of edges whose label will be flipped by the loss on $q_5$.}
\label{fig:interaction}
\end{figure}

Hence a candidate loss pattern $L$ induces an edge-parity label configuration
$$
p_L(i,j)
=
|L\cap A_{ij}| \pmod 2,
\qquad (i,j)\in E,\nonumber
$$
and Eq.~\eqref{eq:predicted_random_checks} can be rewritten as
$$
\mathcal{R}(L) = \{\, i\in V \mid \exists j \text{ s.t. } p_L(i,j)=1 \,\}.
$$

Therefore, the loss inference problem reduces to finding a minimum-weight subset of qubits whose induced edge-label configuration is consistent with the observed random-check set $R_{\mathrm{obs}}$.

If the stabilizer code is LDPC, namely each data qubit participates in at most $O(1)$ stabilizer generators and each generator has weight $O(1)$, then $|E|=O(n)$ and each set $E_q$ has size $O(1)$. In that case, constructing the graph and evaluating $\mathcal{R}(L)$ both require linear time $O(n)$.

The complexity bound above applies only to the primitive operations of constructing the graph (Algorithm~\ref{alg:make_interaction}) and evaluating the map $L\mapsto \mathcal{R}(L)$ for a fixed candidate loss pattern (Algorithm~\ref{alg:evaluate-random-set}).
It does \emph{not} by itself imply a linear-time algorithm for solving the full exact or maximum-likelihood inference problem, since that problem still requires searching over candidate subsets $L$.
For general stabilizer codes, we therefore make no claim here of an efficient exact algorithm for the full optimization problem.
A practical greedy heuristic for the relaxed covering problem, applied to rotated surface codes, is developed in the following sections.

\begin{algorithm}[H]
\caption{Constructing the check-interaction graph}
\label{alg:make_interaction}
\begin{algorithmic}[1]
\State \textbf{Input:} for each data qubit $q$, the set $S(q)$; local Paulis $(s_i)_q$; loss pattern $L$
\State \textbf{Output:} $G=(V,E)$, $\{E_q\}$, and $\mathcal{R}(L)$
\State $V \gets \{1,\dots,m\}$; $E \gets \emptyset$
\Comment{Construct the interaction graph $G$.}
\For{$q=1$ \textbf{to} $n$} 
  \State $E_q \gets \emptyset$
  \ForAll{unordered pairs $(i,j)$ with $i,j\in S(q)$ and $i<j$}
    \If{$(s_i)_q$ and $(s_j)_q$ anticommute locally}
      \State add $(i,j)$ to $E$; add $(i,j)$ to $E_q$
    \EndIf
  \EndFor
\EndFor
\end{algorithmic}
\end{algorithm}

\begin{algorithm}[H]
\caption{Evaluating the map $L \mapsto \mathcal{R}(L)$}
\label{alg:evaluate-random-set}
\begin{algorithmic}[1]
\Require A candidate loss pattern $L \subseteq \{1,\dots,n\}$; a precomputed check-interaction graph $G=(V,E)$; induced edge sets $\{E_q\}_{q=1}^n$
\Ensure The predicted random-check set $\mathcal{R}(L)$
\ForAll{$e \in E$}
    \State $p_L(e) \gets 0$
\EndFor

\ForAll{$q \in L$}
    \ForAll{$e \in E_q$}
        \State $p_L(e) \gets 1 - p_L(e)$
    \EndFor
\EndFor

\State $\mathcal{R}(L) \gets \emptyset$
\ForAll{$e=(i,j) \in E$}
    \If{$p_L(e)=1$}
        \State $\mathcal{R}(L) \gets \mathcal{R}(L) \cup \{i,j\}$
    \EndIf
\EndFor

\State \Return $\mathcal{R}(L)$
\end{algorithmic}
\end{algorithm}

\section{Minimum Set Cover as a relaxation of loss inference}
\label{sec:min_set_cov}
Before introducing the relaxed loss inference problem, we briefly recall the minimum set cover problem~\cite{karp2009reducibility, lemke1971set}. 

\begin{tcolorbox}
\begin{problem}[Minimum Set Cover Problem]\label{prob:min_set_cover}
Let $U$ be a finite universe, and let $\mathcal{S}=\{S_1,\dots,S_N\}$ be a family of subsets of $U$. The minimum set cover problem asks for a subfamily of minimum cardinality whose union covers the entire universe:
\[
\begin{aligned}
\text{minimize}\quad & |\mathcal{C}|,\\
\text{subject to}\quad & \mathcal{C}\subseteq \mathcal{S},\\
& U \subseteq \bigcup_{S\in\mathcal{C}} S.
\end{aligned}
\]
\end{problem}
\end{tcolorbox}

Equivalently, one may assign a binary decision variable $x_k\in\{0,1\}$ to each set $S_k$, where $x_k=1$ means that $S_k$ is selected, and solve
\begin{tcolorbox}
\begin{problem}[Binary Minimum Set Cover Problem]\label{prob:min_set_cover_binary}
\[
\begin{aligned}
\text{minimize}\quad & \sum_{k=1}^N x_k,\\
\text{subject to}\quad & \sum_{k:\,u\in S_k} x_k \ge 1
\qquad \text{for all } u\in U,\\
& x_k\in\{0,1\}
\qquad \text{for all } k=1,\dots,N.
\end{aligned}
\]
\end{problem}
\end{tcolorbox}

That is, every element of the universe must be covered by at least one selected set, while the number of selected sets is minimized.

The exact loss inference problem formulated in Sec.~\ref{sec:problem} provides a natural characterization of detectable qubit loss. However, the exact constraint
\[
R(L)=R_{\mathrm{obs}}
\]
is parity-sensitive and therefore difficult to handle directly. Indeed, for a candidate loss pattern $L \subseteq \{1,\dots,n\}$, the predicted random-check set is defined by
\[
R(L)=\{\, i\in V \mid \exists j \text{ such that } p_L(i,j)=1 \,\},
\]
where
\[
p_L(i,j)=|L\cap A_{ij}| \pmod 2 .
\]
Thus, the map $L \mapsto R(L)$ is non-monotone: adding one more qubit to $L$ may create or cancel odd edge parities modulo $2$, and hence may create or remove elements from $R(L)$. For this reason, the exact inference problem is not itself a minimum set cover problem.

To obtain a more tractable surrogate, we now extract a necessary condition for exact feasibility. Any candidate loss set $L$ that exactly explains the observation $R_{\mathrm{obs}}$ must at least account for every observed random check individually. This suggests a covering viewpoint: each qubit $q$ can potentially explain a certain subset of checks, and any exact explanation must collectively cover all elements of $R_{\mathrm{obs}}$. As shown below, this observation leads to a natural minimum set cover relaxation of the exact loss inference problem.

Recall from Sec.~\ref{sec:problem} that the check-interaction graph is $G=(V,E)$, where each vertex corresponds to a stabilizer check, and each data qubit $q\in\{1,\dots,n\}$ induces the edge subset
$$
E_q := \{(i,j)\in E \mid q\in A_{ij}\}.
$$
That is, losing qubit $q$ flips the parity on every edge in $E_q$.

For each qubit $q$, define the associated check set
\[
B_q := \{\, i\in V \mid \exists j\in V \text{ such that } (i,j)\in E_q \,\}.
\]
Equivalently, $B_q$ is the set of checks incident to at least one edge whose parity can be flipped by the loss of qubit $q$.

The key observation is that any exact explanation of the observed random-check set must cover all observed random checks by these sets $\{B_q\}$.

\begin{proposition}
Let $L \subseteq \{1,\dots,n\}$ be feasible for the exact loss inference problem, namely suppose that $R(L)=R_{\mathrm{obs}}$.

Then
$$
R_{\mathrm{obs}} \subseteq \bigcup_{q\in L} B_q.
$$
\label{prop:bq}
\end{proposition}

\begin{proof}
Take any $i\in R_{\mathrm{obs}}$. Since $R(L)=R_{\mathrm{obs}}$, we have $i\in R(L)$. By the definition of $R(L)$, there exists some $j$ such that
\[
p_L(i,j)=|L\cap A_{ij}| \equiv 1 \pmod 2.
\]
Hence $L\cap A_{ij}\neq \varnothing$, so there exists at least one qubit $q\in L\cap A_{ij}$. By the definition of $E_q$, this implies $(i,j)\in E_q$, and therefore $i\in B_q$. Since $q\in L$, we conclude that
\[
i\in \bigcup_{q\in L} B_q.
\]
Because $i\in R_{\mathrm{obs}}$ was arbitrary, the claim follows.
\end{proof}

Proposition~\ref{prop:bq} shows that exact feasibility implies a covering constraint. This yields a natural relaxation of the exact loss inference problem.

\begin{corollary}
Define
\[
\mathcal{F}_{\mathrm{exact}}
:=
\{\,L\subseteq \{1,\dots,n\}\mid R(L)=R_{\mathrm{obs}}\,\},
\]
and
\[
\mathcal{F}_{\mathrm{cov}}
:=
\left\{
L\subseteq \{1,\dots,n\}
\ \middle|\
R_{\mathrm{obs}}\subseteq \bigcup_{q\in L} B_q
\right\}.
\]
Then
\[
\mathcal{F}_{\mathrm{exact}} \subseteq \mathcal{F}_{\mathrm{cov}}.
\]
Consequently, if both problems minimize $|L|$, then
\[
\min_{L\in\mathcal{F}_{\mathrm{cov}}}|L|
\le
\min_{L\in\mathcal{F}_{\mathrm{exact}}}|L|.
\]
\end{corollary}
Thus, the optimum of the relaxed problem provides a lower bound on the optimum of the exact loss inference problem.

Motivated by this inclusion, we define the relaxed loss inference problem as
\begin{tcolorbox}
\begin{problem}[Relaxed Loss Inference Problem]\label{prob:relaxed_loss_inf}

\[
\begin{aligned}
\text{minimize}\quad & |L|,\\
\text{subject to}\quad & R_{\mathrm{obs}} \subseteq \bigcup_{q\in L} B_q,\\
& L\subseteq \{1,\dots,n\}.
\end{aligned}
\]
This is a minimum set cover problem with universe $R_{\mathrm{obs}}$ and set family
\[
\{\, B_q\cap R_{\mathrm{obs}} \,\}_{q=1}^n.
\]

\end{problem}
\end{tcolorbox}

In other words, each qubit $q$ is regarded as a candidate set that covers the observed random checks that can be affected by the loss of $q$, and the goal is to explain all observed random checks using as few lost-qubit candidates as possible.

It is important to emphasize that this formulation is a relaxation, not an equivalent reformulation, of the exact inference problem. The exact problem requires the full parity-consistency condition $R(L)=R_{\mathrm{obs}}$, where odd edge parities generated by different lost qubits may interfere through cancellation modulo $2$. By contrast, the set cover formulation retains only the necessary condition that every observed random check be incident to at least one qubit whose loss can flip a relevant edge parity. Therefore, a feasible solution of the set cover problem need not be feasible for the exact problem. The benefit of the relaxation is that it replaces a parity-sensitive combinatorial matching problem by a monotone covering problem, which admits standard approximation algorithms.

\section{Heuristic Algorithms for Loss Inference}
\label{sec:heuristic}
\subsection{Greedy Minimum Set Cover Algorithm}
\label{subsec:greedy}

The set-cover relaxation introduced in Sec.~\ref{sec:min_set_cov} remains NP-hard in general~\cite{karp2009reducibility}. We therefore adopt a standard greedy procedure and use its output as a heuristic inferred loss set. The basic idea is to iteratively choose a qubit whose associated covering set explains as many currently uncovered random checks as possible.

\begin{algorithm}[H]
\caption{Greedy heuristic for relaxed loss inference}
\label{alg:greedy_loss_inference}
\begin{algorithmic}[1]
\Require Observed random-check set $R_{\mathrm{obs}}$; covering sets $\{B_q\}_{q=1}^n$
\Ensure Inferred loss set $\hat{L}$
\State $U \gets R_{\mathrm{obs}}$ \Comment{set of currently uncovered random checks}
\State $\hat{L} \gets \varnothing$
\While{$U \neq \varnothing$}
    \State choose
    \[
    q^\star \in \arg\max_{q\in\{1,\dots,n\}\setminus \hat{L}} |B_q \cap U|
    \]
    \State $\hat{L} \gets \hat{L} \cup \{q^\star\}$
    \State $U \gets U \setminus B_{q^\star}$
\EndWhile
\State \Return $\hat{L}$
\end{algorithmic}
\end{algorithm}

By construction, the output $\hat{L}$ is feasible for the set-cover relaxation, namely
\[
R_{\mathrm{obs}} \subseteq \bigcup_{q\in \hat{L}} B_q.
\]
Indeed, the algorithm terminates only when all elements of $R_{\mathrm{obs}}$ have been covered. It is important to emphasize, however, that the greedy procedure is an approximation algorithm for the relaxed covering problem, not an exact solver for the parity-sensitive constraint
\[
R(L)=R_{\mathrm{obs}}.
\]

Nevertheless, the greedy rule is attractive for practical inference because it is simple and depends only on the local incidence structure of the family $\{B_q\}_{q=1}^n$. Moreover, since the relaxed problem is a minimum set cover problem, the standard approximation analysis of greedy set cover applies directly: the size of the greedy solution is within the usual logarithmic factor of the optimum of the relaxed problem~\cite{slavik1996tight}.

A naive implementation recomputes the gain $|B_q\cap U|$ for every qubit at every iteration, giving $O(n)$ work per selected qubit and hence $O(n\,|\hat{L}|)$ overall; this is the implementation we use throughout Sec.~\ref{sec:eval}. Because each set $B_q$ is local and typically small in sparse code families such as the rotated surface code, we expect that a bucket-queue implementation analogous to those used in Union-Find peeling decoders~\cite{delfosse2021almost,delfosse2020linear} could reduce this to time linear in $|E|$; working out this construction in detail is left to future work.

\subsection{Signature-Recall Loss Inference}
\label{subsec:hc}

The greedy algorithm of Sec.~\ref{subsec:greedy} addresses the covering constraint $R_{\mathrm{obs}} \subseteq \bigcup_{q \in \hat{L}} B_q$, but does not require each
selected qubit to individually match the observed pattern closely.
In a multi-loss scenario, a qubit may be chosen simply because it covers a
few remaining uncovered checks, even if most of its associated covering set
$B_q$ lies outside $R_{\mathrm{obs}}$, introducing false positives.

We propose a complementary approach, \emph{signature-recall inference},
which retains only qubits whose single-loss signature is in strong individual
agreement with $R_{\mathrm{obs}}$.
The key ingredient is a precomputed single-loss signature table: for each
data qubit $q \in \{1, \ldots, n\}$, we simulate the circuit with a single
loss at $q$ and record the resulting random-check set under the classifier of
Sec.~\ref{subsec:classification},
\begin{equation}
    \hat{R}_q := \mathrm{Classify}\!\left(\mathrm{Syndrome}(L = \{q\}),\, W,\, \tau\right).
\end{equation}
Because only one qubit is lost, $\hat{R}_q$ captures the unambiguous,
noise-averaged signature of a single loss event.
The table $\{\hat{R}_q\}_{q=1}^n$ is computed offline and reused across trials.

Given $R_{\mathrm{obs}}$, we score each qubit by the
\emph{signature-recall} of $q$'s expected signature against the observation:
\begin{equation}
    s(q) := \frac{\lvert R_{\mathrm{obs}} \cap \hat{R}_q \rvert}
                 {\lvert \hat{R}_q \rvert},
\end{equation}
with $s(q) = 0$ if $\hat{R}_q = \varnothing$.
The score measures what fraction of $q$'s expected random checks actually
fired.
The denominator $|\hat{R}_q|$ is determined solely by the local neighbourhood of qubit $q$
(approximately $2$--$4$ for $W=1$) and is independent of the code
distance $d$.
Background Pauli noise outside $\hat{R}_q$ therefore contributes nothing to
the denominator, making the score robust across all $d$ and $p_{\mathrm{Pauli}}$
values.

\paragraph*{Analytical threshold derivation.}
When qubit $q$ is truly lost, the observed syndrome pattern matches
$\hat{R}_q$ almost exactly, so the true-positive score is $\approx 1$.
A false-positive candidate $q'$ (not lost) scores
\begin{equation}
    s_{\mathrm{FP}}(q' \mid q\text{ lost})
        = \frac{|\hat{R}_q \cap \hat{R}_{q'}|}{|\hat{R}_{q'}|},
\end{equation}
which measures the fraction of $q$'s syndrome that overlaps with $q'$'s
expected signature.
Two classes of false positives arise.

\emph{(A) Removable false positives} ($\hat{R}_{q'} \not\subseteq \hat{R}_q$):
In a $d \times d$ rotated surface code, each data qubit touches at most $4$
auxiliary qubits (bulk) or fewer (boundary/corner).
Two distinct qubits share at most $2$ auxiliary-round events in their $W=1$
signatures, and the minimum signature size of a non-subset candidate is $3$.
This caps the false-positive score at $2/3 \approx 0.667$.
Numerical verification for $d = 3, 5, 7, 9, 11$ confirms that no
non-subset pair exceeds this bound, and the bound is tight (some
edge-qubit pairs achieve exactly $2/3$).

\emph{(B) Irreducible false positives} ($\hat{R}_{q'} \subseteq \hat{R}_q$):
Boundary or corner qubits have fewer adjacent auxiliary qubits, so their entire
signature can be a proper subset of a neighbouring bulk qubit's signature.
These always score $1$ and cannot be eliminated by thresholding.
For a $d \times d$ surface code with $W=1$, exactly $4$ such irreducible
pairs exist regardless of $d$, corresponding to the $4$ corner regions of the
lattice.

The resulting score distribution is\\
$s \in [0,\, \tfrac{2}{3}] \;(\text{removable FP})$,
$s \in (\tfrac{2}{3},\, 1) \;(\text{empty gap})$,\\
$s = 1 \;(\text{TP and irreducible FP})$.
Setting the threshold just above the removable-FP ceiling eliminates all
removable false positives while retaining all true positives.
Concretely, we use
\begin{equation}
    \theta^* = \tfrac{2}{3} + \delta = 0.717,
    \qquad \delta = 0.05,
\end{equation}
and define the inferred set as
\begin{equation}
    \hat{L}_{\mathrm{hc}}
        := \bigl\{ q \in \{1, \ldots, n\} \mid s(q) \ge \theta^* \bigr\}.
\end{equation}

This threshold is analytically derived and independent of $d$, $p_{\mathrm{Pauli}}$, and the hardware noise model, because signature-recall scores depend only on the lattice geometry.

\begin{algorithm}[H]
\caption{Signature-recall heuristic for high-confidence loss inference}
\begin{algorithmic}[1]
\Require 
Observed random-check set $R_{\mathrm{obs}}$;
         signature table $\{\hat{R}_q\}_{q=1}^n$;
         threshold $\theta^* = 2/3 + \delta$ with default $\delta = 0.05$
\Ensure Inferred loss set $\hat{L}_{\mathrm{hc}}$

\State $\hat{L}_{\mathrm{hc}} \gets \varnothing$

\For{$q = 1$ \textbf{to} $n$}
    \If{$\hat{R}_q = \varnothing$}
        \State $s(q) \gets 0$
    \Else
        \State $s(q) \gets |R_{\mathrm{obs}} \cap \hat{R}_q| / |\hat{R}_q|$
    \EndIf

    \If{$s(q) \ge \theta^*$}
        \State $\hat{L}_{\mathrm{hc}} \gets \hat{L}_{\mathrm{hc}} \cup \{q\}$
    \EndIf
\EndFor
\State \Return $\hat{L}_{\mathrm{hc}}$
\end{algorithmic}
\label{alg:sig-recall}
\end{algorithm}


\section{Evaluation of heuristic loss inference}
\label{sec:eval}
We have simulated the performance of loss correction protocols P0 to P5 on the rotated surface code~\cite{bombin2007optimal} as summarized in Table~\ref{table:protocols}. P0 is the baseline of the logical error rate without any loss correction. P1 is the idealized case of applying the LDU without any error. It immediately reloads the lost qubit without any reloading latency. This corresponds to the erasure channel and gives a rough lower bound of the logical error rate. P2 is the noisy LDU reload: the LDU circuit detects losses subject to a detection delay $d_{\mathrm{LDU}}$ and measurement error probability.
P3 is the high-confidence inferred reload of Sec.~\ref{subsec:hc}: syndrome statistics are collected over $W$ rounds and only qubits with Jaccard score $s(q)\geq\theta$ are reloaded, favoring precision over recall.
P4 is the naive greedy reload: the full greedy set-cover candidate list from Algorithm~\ref{alg:greedy_loss_inference} is reloaded, achieving high recall at the cost of more false positives. P5, P4$\to$adaptive LDU, uses P4 first to detect suspicious loss errors, then applies LDUs only to them, reducing the number of LDU applications. For all protocols, we use the same decoding algorithm, which clusters the errors by Union-Find~\cite{delfosse2021almost}, which can treat both loss and Pauli errors, and then applies Gaussian elimination as a subroutine instead of the peeling algorithm~\cite{delfosse2020linear}. This is because the peeling algorithm is designed for the code capacity model and has not been generalized for the detector model and hyperedge because of error propagation. 
The Gaussian elimination on a cluster of $|L|$ erased/flagged qubits costs $O(|L|^3)$. For LDPC codes, clusters are typically $O(1)$ in size and this cost is small, but clusters can grow to $O(n)$ under high loss rates, in which case the per-shot cost can reach $O(n^3)$ in the worst case; this decoder can be replaced by a more sophisticated one with better worst-case scaling in future work.
We use the reloading plus stabilizer measurement scheme introduced in \cite{wu2022erasure, nishio2025multiplexed} for correcting erasure errors rather than the super-stabilizer-based approach~\cite{stace2009thresholds, barrett2010fault, kobayashi2025erasure}, which does not recover the original code distance.

\begin{table}[t]
\centering
\caption{Summary of loss-correction protocols.
$W$: inference window (rounds); $l$: reload latency (rounds);
$d_{\mathrm{LDU}}$: LDU detection delay (rounds).}
\begin{tabular}{cll}
\hline
\textbf{Protocol} & \textbf{Detection method} & \textbf{Reload delay} \\
\hline
P0 & No correction                              & ---                   \\
P1 & Heralded erasure                  & $l$                   \\
P2 & Noisy LDU                       & $d_{\mathrm{LDU}} + l$ \\
P3 & Signature-recall inference        & $W + l$               \\
P4 & Greedy set cover inference                  & $W + l$               \\
P5 & P4 $\to$ adaptive LDU     & $W + d_{\mathrm{LDU}} + l$ \\
\hline
\label{table:protocols}
\end{tabular}
\end{table}

\begin{figure*}[htbp]
    \centering
    \includegraphics[width=\linewidth]{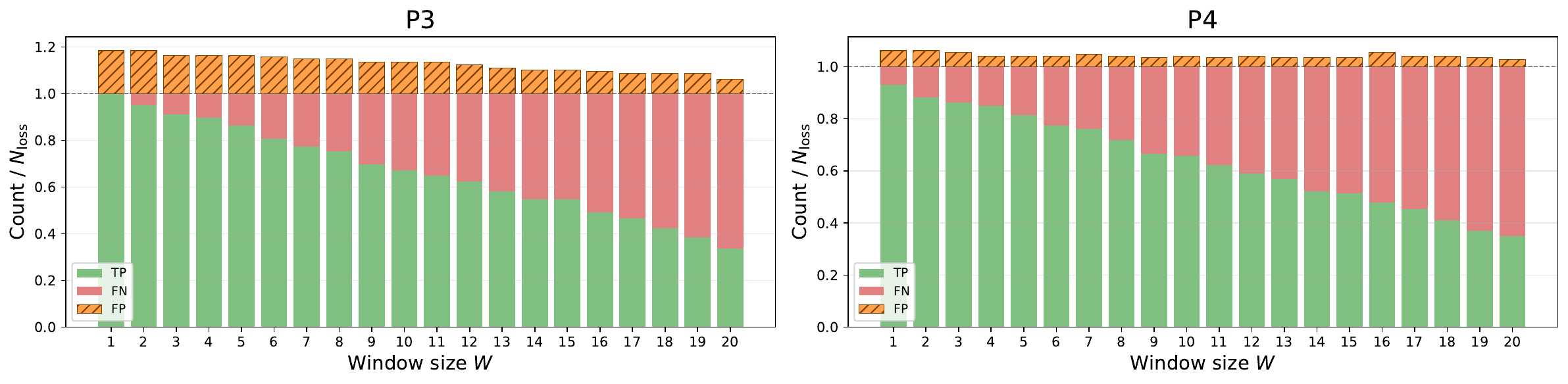}
    \caption{Inference outcome breakdown of loss inference protocols P3 and P4 versus the inference time window size $W$. Code distance $d=7$, $Cum(p_{\rm loss})=0.01$, gate error rates and SPAM error rate are based on trapped ion hardware shown in Table.~\ref{tab:realistic_noise}.}
    \label{fig:W_sweep}
\end{figure*}

\subsection{Loss inference classification performance}
\label{subsec:classification_performance}

We evaluated the classification performance of the loss inference
protocols P3 and P4.
We evaluated accuracy using a confusion matrix approach, in which the inference outcome
for each candidate qubit is classified as a true positive (TP), true negative (TN), false positive (FP), or false negative (FN). These results show that $W=1$ consistently gives the best recall, $\mathrm{TP}/(\mathrm{TP}+\mathrm{FN})$.
Because multiple loss patterns can in principle yield the same observed random-check set, the inference algorithm can underestimate the number of losses when this ambiguity arises.

We evaluated the logical error rate of the protocols based on a realistic
noise model\footnote{Each row combines error rates reported in separate experiments (and, for the neutral-atom row, on different atomic species/apparatus). No single demonstrated system has achieved all three values simultaneously, so the composite should be read as an optimistic near-term target rather than an end-to-end rate measured on one device.} as shown in Table~\ref{tab:realistic_noise}.

\begin{table}[hb]
\begin{tabular}{lllll}
\cline{1-4}
\multicolumn{1}{|l|}{}             & \multicolumn{1}{l|}{1 Qubit Gate} & \multicolumn{1}{l|}{2 Qubit Gate} & \multicolumn{1}{l|}{SPAM} &  \\ \cline{1-4}
\multicolumn{1}{|l|}{Trapped Ion}  & \multicolumn{1}{l|}{$1.5 \times 10^{-7}$~\cite{smith2025single}}   & \multicolumn{1}{l|}{$8.4\times10^{-5}$\cite{hughes2025trapped}}   & \multicolumn{1}{l|}{$5\times 10^{-6}$~\cite{sotirova2024high}}     &  \\ \cline{1-4}
\multicolumn{1}{|l|}{Neutral atom} & \multicolumn{1}{l|}{$4.7\times10^{-5}$~\cite{sheng2018high}}   & \multicolumn{1}{l|}{$1.46\times10^{-3}$~\cite{kiefer2026protected}}   & \multicolumn{1}{l|}{$1.5\times10^{-4}$~\cite{wang2025ultrafast}}     &  \\ \cline{1-4}
\end{tabular}
\caption{Error rates from state-of-the-art hardware experiments used for the logical error rate simulations.}
\label{tab:realistic_noise}
\end{table}

We first evaluated the classification performance of the loss inference
protocols P3 and P4 versus the window size $W$ as shown in Fig.~\ref{fig:W_sweep}. 
These results show that $W=1$ consistently gives the best recall, and that recall degrades roughly linearly as $W$ increases. This trend follows directly from the definition of $q_i(W)$ in Eq.~\eqref{eq:windowed-score}. Because reload is disabled in this benchmark, a check made anticommuting by a loss at round $r_{\mathrm{loss}}$ stays random for the remainder of the trace, i.e.\ over the interval
$[r_{\mathrm{loss}}, T)$ of length $T-r_{\mathrm{loss}}$. The score
$q_i(W)$ can reach its saturated value $\pi_{i,r}\approx 1/2$ only if some length-$W$ window fits entirely inside this active interval; since no admissible window may extend past round $T$, this requires $T - r_{\mathrm{loss}} \geq W$. 
When the loss occurs within the final $W-1$ rounds of the trace (i.e.\ $r_{\mathrm{loss}} > T-W$), every admissible window is forced to include some deterministic, pre-loss rounds.
Pinning the window against the trace boundary, which is the best achievable placement in that case, gives a diluted score
$$
  q_i(W) \;\approx\; \frac{\min(T-r_{\mathrm{loss}},\,W)}{W}\cdot\frac12,
$$
so detection succeeds only if $T - r_{\mathrm{loss}} \geq 2\tau W$. Since
$r_{\mathrm{loss}}$ is approximately uniformly distributed over the
$T$-round trace, the resulting recall follows
$$
  \mathrm{recall}(W) \;\approx\; \max\!\left(0,\; 1 - \frac{2\tau W}{T}\right),
$$
which for $d=7$ ($T=21$, $\tau=0.35$) predicts $0.97$, $0.67$, and $0.33$ at
$W=1$, $10$, $20$ respectively — closely matching the observed P3 recall of
$1.00$, $0.65$, and $0.32$ in Fig.~\ref{fig:W_sweep}. $W=1$ is immune to this effect
because a single-round window is never forced to straddle the loss onset,
so it always attains the peak flip rate regardless of when the loss
occurred within the trace.

The false-positive count shrinks slightly with $W$ for a complementary reason: spurious threshold crossings driven by background noise are smoothed out by the same averaging that dilutes genuine signal. Because the background error rates used here are far below $\tau$, this effect is much weaker than the recall loss, so overall classification performance is best at $W=1$ throughout Sec.~\ref{sec:eval}. Because multiple loss patterns can in principle yield the same observed random-check set, the inference algorithm can also underestimate the
number of losses when this ambiguity arises, independent of the effect
above.

Next, Fig.~\ref{fig:cm_d_sweep} breaks down the inference outcomes as a function of code distance $d$. The experiment fixes the cumulative per-qubit loss probability over the $3d$-round circuit,
$$
  \mathrm{Cum}(p_{\mathrm{loss}}) \;:=\; 1 - (1-p_{\mathrm{loss}})^{3d} \;=\; 0.01,
$$
where $p_{\mathrm{loss}}$ is the per-round loss probability.
So the expected number of concurrent losses is
$\mathbb{E}[N_{\mathrm{loss}}] \approx \mathrm{Cum}(p_{\mathrm{loss}})\, d^2$,
growing quadratically with code distance.

P3 accepts a candidate qubit $q$ using the \emph{signature-recall} score of
Sec.~\ref{subsec:hc},
$$
  s(q) = \frac{|R_{\mathrm{obs}} \cap \hat R_q|}{|\hat R_q|} \;\ge\; \theta^\ast = 0.717.
$$
Because the denominator $|\hat R_q|$
depends only on the local neighbourhood of $q$, $s(q)$ for a truly lost qubit stays close to $1$ regardless of how many
\emph{other} qubits are simultaneously lost, as long as $\hat R_q \subseteq
R_{\mathrm{obs}}$. It holds unless the anticommutation contributions of the other losses happen to cancel exactly on one of $q$'s own checks, an event we do not observe at the loss rates considered here. Table~\ref{tab:p3p4-recall-precision} confirms this: P3's recall is $1.000$ at every tested $d$, from $d=3$ ($\mathbb{E}[N_{\mathrm{loss}}]
\approx 0.09$) through $d=9$ ($\mathbb{E}[N_{\mathrm{loss}}] \approx 0.75$).

\begin{figure*}[htbp]
    \centering
    \includegraphics[width=0.8\linewidth]{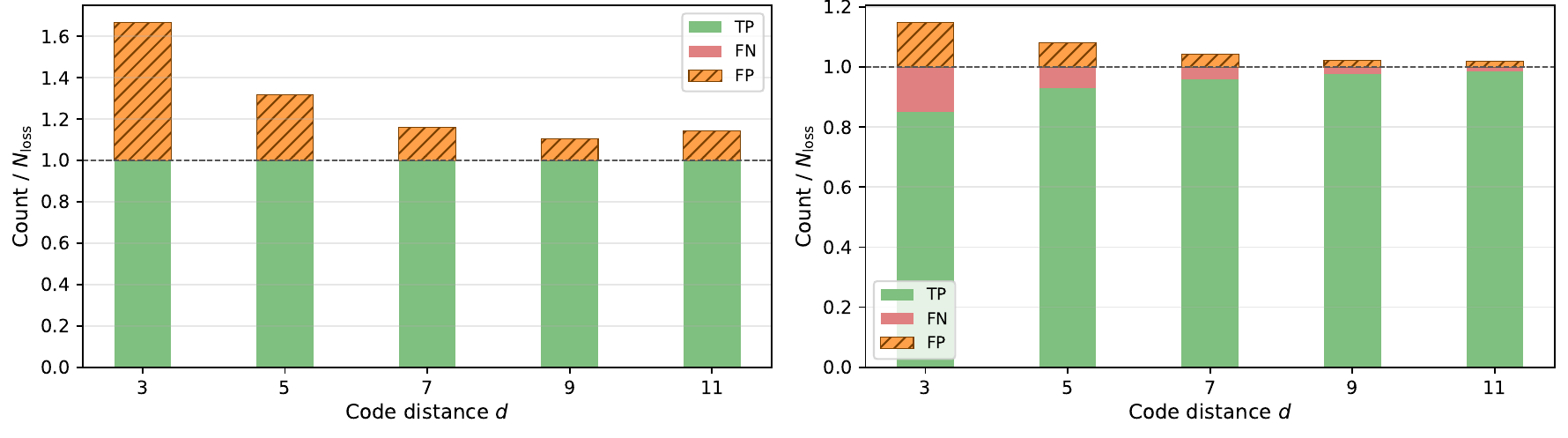}
    \caption{Inference outcome breakdown of loss inference protocols P3 (left) and P4 (right) versus the code distance $d$. It is using the trapped ion-based 1Q/2Q/SPAM errors and $Cum(p_{\rm loss})=0.01$. $W=1$, $300$ trials.}
    \label{fig:cm_d_sweep}
\end{figure*}

\begin{figure*}[htbp]
    \centering
    \includegraphics[width=\linewidth]{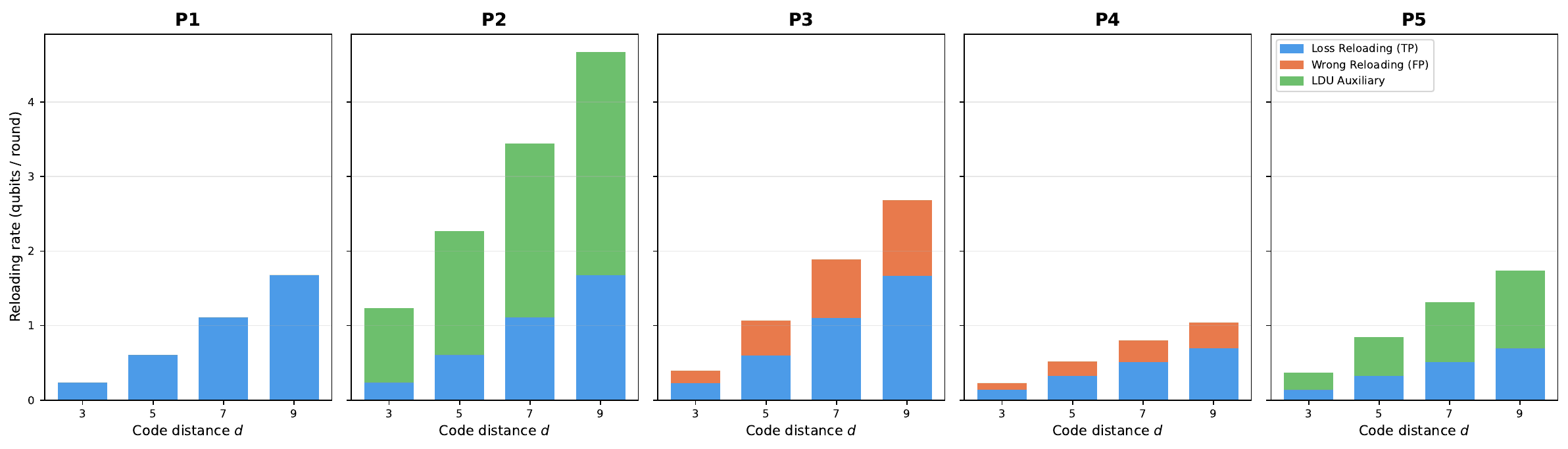}
    \caption{Qubit consumption rate for protocol P1 to P5. The blue stack shows the reloading overhead, and the green shows the auxiliary qubits required for LDUs. The orange one is loss inference-specific overhead, which is additional reloading based on false positive decisions. $p_{\rm loss}  = 0.03$, $p_{\rm pauli} = 1e-3$}
    \label{fig:reload_ldu_rate}
\end{figure*}

\begin{table}[htbp]
\centering
\caption{P3 and P4 inference accuracy versus code distance $d$, trapped-ion
hardware parameters (Table~\ref{tab:realistic_noise}), $Cum(p_{loss})=0.01$, $W=1$, $\tau=0.35$, 300 trials per point.}
\label{tab:p3p4-recall-precision}
\begin{tabular}{ccccc}
\hline\hline
$d$ & $\mathbb{E}[N_{\mathrm{loss}}]$ & P3 recall / precision & P4 recall / precision \\
\hline
3 & 0.09 & 1.000 / 0.600 & 0.852 / 0.852 \\
5 & 0.33 & 1.000 / 0.760 & 0.929 / 0.919 \\
7 & 0.56 & 1.000 / 0.861 & 0.958 / 0.958 \\
9 & 0.75 & 1.000 / 0.907 & 0.978 / 0.978 \\
\hline\hline
\end{tabular}
\end{table}

What does change with $d$ is P3's \emph{precision}, which rises from $0.600$
at $d=3$ to $0.907$ at $d=9$. This trend has a direct explanation in terms of
the two false-positive classes already identified in
Sec.~\ref{subsec:hc}: for a fixed hardware noise level, the \emph{count} of false-positive candidates per trial is dominated by a small, essentially $d$-independent number of irreducible or near-irreducible overlap events (Sec.~\ref{subsec:hc} identifies exactly four irreducible-overlap qubit pairs per lattice, regardless of $d$), whereas the number of true positives $\mathrm{TP}=\mathbb{E}[N_{\mathrm{loss}}]$ grows as $\mathrm{Cum}(p_{\mathrm{loss}})\, d^2$. A roughly constant false-positive count divided by a quadratically growing true-positive count drives precision toward $1$ as
$d$ increases, without any change in P3's recall.

P4, by contrast, uses the greedy set-cover heuristic of
Sec.~\ref{subsec:greedy}, which has no hard score threshold: it
iteratively selects the candidate whose signature covers the most
unexplained random checks, removes that signature from the residual, and
repeats. Table~\ref{tab:p3p4-recall-precision} shows P4's recall rising from
$0.852$ at $d=3$ to $0.978$ at $d=9$, tracking its precision almost exactly
at every $d$ — i.e.\ P4's false-negative and false-positive counts remain
comparable in magnitude as $d$ grows, unlike P3, which trades a
recall-favouring threshold for a shrinking but nonzero excess of false
positives at small $d$.

In summary, P3 and P4 do not diverge as $d$ grows: both converge toward
high-accuracy inference, but along different paths. P3, thresholded on
signature-recall, never misses a genuine loss in the regime tested here
(recall $\equiv 1$) and pays for this with extra reload operations that
become negligible, in relative terms, once $d$ is large enough that
$N_{\mathrm{loss}}$ dominates the small, roughly constant false-positive
count. P4 trades a small amount of recall at low $d$ for higher precision,
and gains recall as concurrent losses provide more structure for the greedy
decomposition to exploit, so that by $d=9$ its recall and precision are both
within a few percent of P3's.

These classification differences translate directly into qubit reload demand. Fig.~\ref{fig:reload_ldu_rate} compares the qubit consumption\footnote{This rate assumes that one can only apply destructive measurement. For the systems capable of non-destructive measurement, the LDU auxiliary might not require new qubits.} per code cycle for protocols P1–P5 as a function of the code distance $d$.
The qubit consumption is decomposed into three contributions: qubit reloading after actual loss events, unnecessary reloading caused by false-positive loss inference, and auxiliary qubits required for LDUs. 
As expected, LDU-based protocols (P2 and P5) require a substantial number of auxiliary qubits, and this overhead increases with the code distance. In contrast, the inference-based protocols (P3 and P4) eliminate the need for LDU auxiliary qubits entirely. Although false-positive inference introduces additional unnecessary reloads, their contribution remains significantly smaller than the auxiliary-qubit overhead required by LDU-based approaches under the simulated conditions. These results indicate that loss inference can substantially reduce the qubit resource requirements while maintaining comparable functionality, which is evaluated later.

\subsection{Logical error rate performance}
\label{subsec:logical_error_rate_pauli_loss_only}

We next evaluate how the accuracy of loss inference translates into
logical-level performance. Fig.~\ref{fig:LER_loss} and \ref{fig:LER_pauli} show the logical error rate
(LER) of protocols P0--P5 as functions of the physical loss probability per round
and the circuit-level Pauli error probability, respectively. 
The LER is evaluated after $3d$ rounds of syndrome extraction.

Fig.~\ref{fig:LER_loss} isolates the effect of loss errors by setting the background
Pauli error probability to a negligibly small value $10^{-10}$. 
Without loss mitigation (P0), the logical error rate rapidly increases as the loss probability grows. The ideal erasure benchmark (P1) achieves the lowest logical error rate over the entire range. Among practical protocols P2-P5, P2 outperforms the others for a high loss rate more than $10^{-2}$ whereas $P3$ and $P4$ outperform P2/P5 for a smaller loss rate region.
This indicates that avoiding noisy LDU operations can compensate for the imperfect accuracy of loss inference. At sufficiently high loss rates, the performance of all inference-based protocols gradually worse than the LDU-based protocols as multiple simultaneous loss events become increasingly difficult to infer correctly.

\begin{figure}[htbp]
    \centering
    \includegraphics[width=\linewidth]{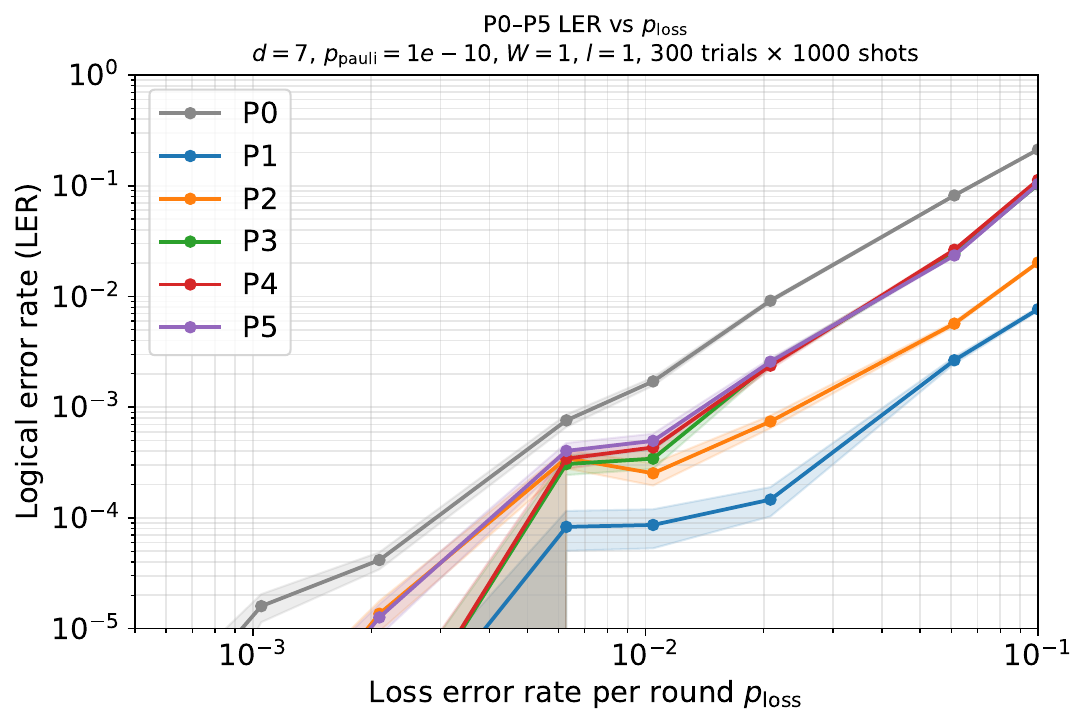}
    \caption{Logical error rate versus the loss error rate with the fixed $p_{\rm pauli} = 10^{-10}$. 
    }
    \label{fig:LER_loss}
\end{figure}

Fig.~\ref{fig:LER_pauli} plots the logical error rate as a function of the physical Pauli error rate while keeping the loss probability $0$. In this regime, all protocols exhibit nearly identical performance because loss events are extremely rare and Pauli errors dominate the logical failure probability. Consequently, the advantage of loss inference becomes negligible, and improving the underlying gate fidelity remains the primary requirement for reducing logical errors. This result confirms that the proposed inference protocol introduces essentially no performance degradation when loss errors are absent.

\begin{figure}[htbp]
    \centering
    \includegraphics[width=\linewidth]{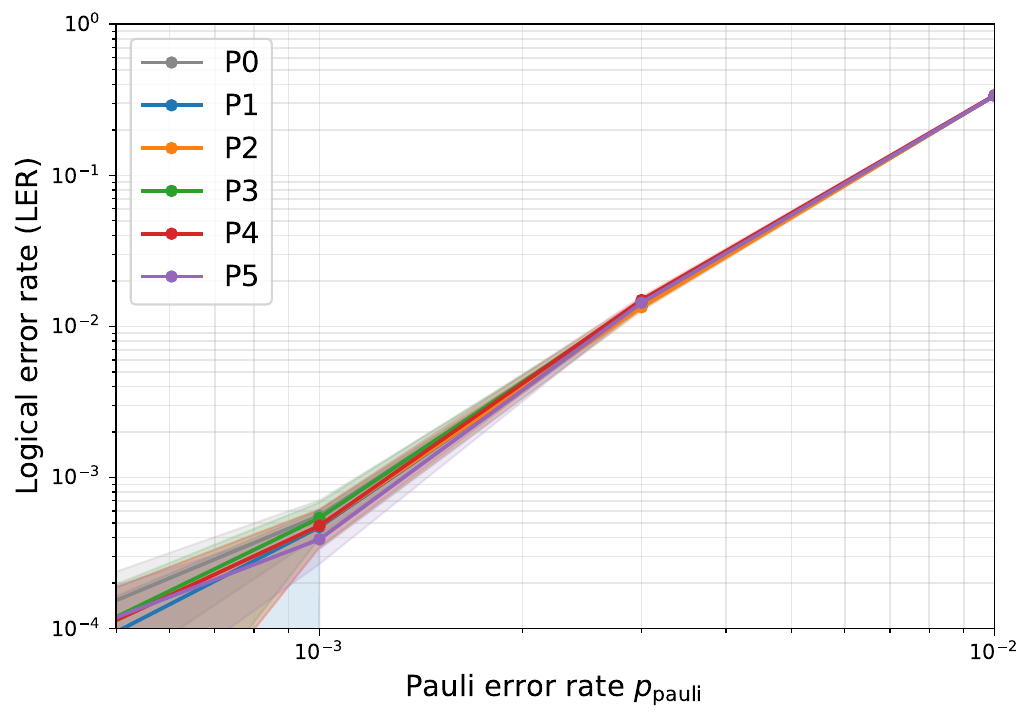}
    \caption{Logical error rate of protocols versus the Pauli error rate for the case without loss error. $d=7$, $W=1$, $l=1$, $10^{5}$ shots per point.}
    \label{fig:LER_pauli}
\end{figure}

\begin{figure*}[htbp]
\subfloat{%
  \includegraphics[clip,width=\columnwidth]{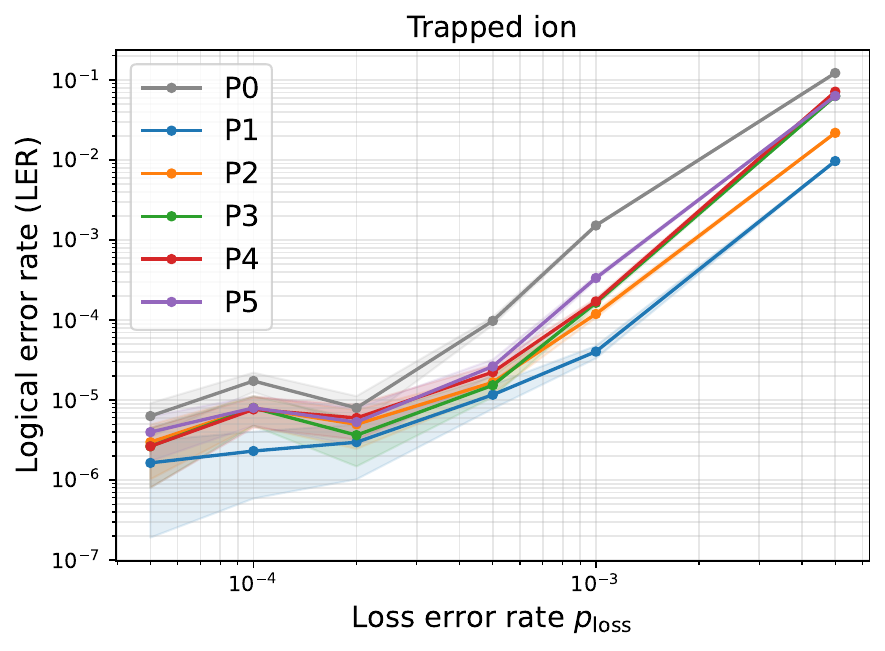}
}
\subfloat{%
  \includegraphics[clip,width=\columnwidth]{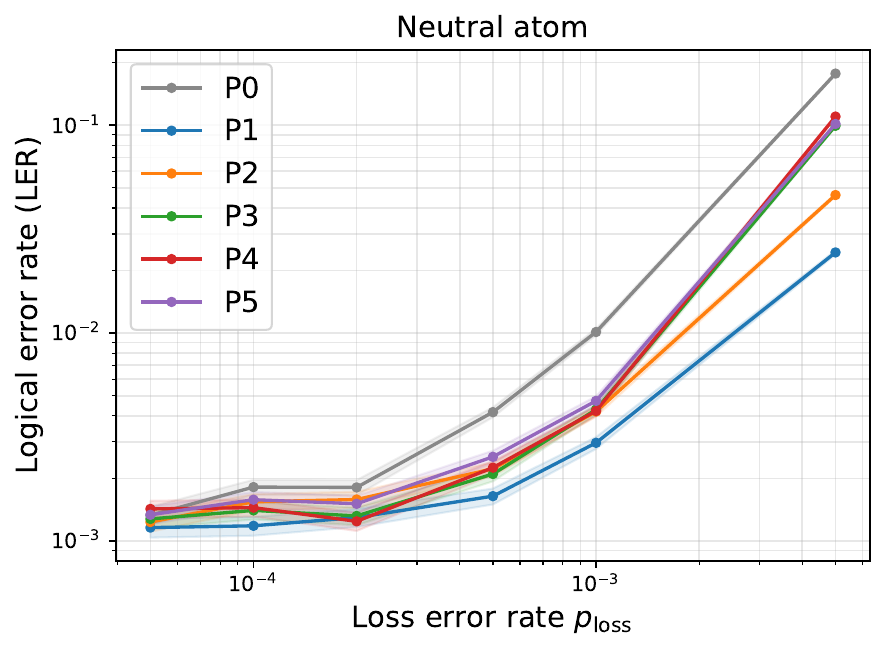}
}
\caption{Logical error rates of protocols versus the loss error rate with a realistic noise model for the trapped ion and the neutral atom. $d=7$, $300$ trials for loss errors and $10000/1000$ shots for ions/atoms.}
\label{fig:realistic}
\end{figure*}

\subsection{LER on realistic error rates}

Fig.~\ref{fig:realistic} compares the logical error rates of protocols P0–P5 under the trapped-ion and neutral-atom noise parameters listed in Table~\ref{tab:realistic_noise}. 

On both platforms, P0 worsens far faster than a linear scaling with $p_{\rm loss}$ would predict, reflecting the rising probability of multiple simultaneous, uncorrected losses when no mitigation is applied.

Up to roughly $p_{\rm loss} \approx 5e-4$, the inference-based protocols (P3/P4) match or slightly outperform the noisy-LDU protocol (P2). In this regime, statistical inference from syndrome data alone achieves LDU-level (or better) loss mitigation without requiring a physical Leakage Detection Unit with space-time overheads. This advantage narrows and then reverses as $p_{\rm loss}$ rises further. The two become roughly tied around $p_{\rm loss}\approx 1e-3$, and by $p_{\rm loss} \approx 5e-3$, P2 pulls ahead of P3/P4 by roughly a factor of 2 on both platforms. This is a favorable result for the proposed inference-based (LDU-free) approach in the realistic, low-to-moderate loss-rate regime, but it also marks a clear validity boundary.

The crossover occurs at nearly the same $p_{\rm loss}$ on both hardware platforms, which suggests it is a structural consequence of the inference window ($W = 1$) rather than a hardware-specific effect. At low $p_{\rm loss}$, individual loss events are well-separated in time and the syndrome signature within a single window unambiguously identifies the lost qubit, so inference matches LDU performance without the LDU overhead. As $p_{\rm loss}$ grows, however, the probability of multiple losses overlapping within a single window increases, degrading the signature-based disambiguation used by high-confidence filtering and greedy set cover — while LDU-based detection (P2), which measures loss directly, is unaffected by this ambiguity and remains robust.

P5 (adaptive LDU) sits between P2 and P3/P4 on both platforms at high $p_{\rm loss}$, suggesting the adaptive triggering partially inherits the benefits of both underlying strategies.

For the case of trapped ion, the baseline (P0) failure rate is extremely low at small $p_{\rm loss}$ ($~1e-6$) but degrades dramatically once $p_{\rm loss}$ exceeds $~1e-3$, rising roughly 80-fold (0.00152 → 0.122) for only a 5× increase in loss rate (1e-3 → 5e-3). This steep, super-linear blow-up indicates that once qubit loss becomes non-negligible, uncorrected multi-loss events start dominating logical failures far out of proportion to their probability — leaving essentially no safety margin without mitigation at higher loss rates, even though the underlying gate error rate is low.

It also shows that P3 slightly beats P2 at $p_{\rm loss}=5e-4$, but by for $p_{\rm loss}=1e-3$ P2 is clearly ahead, widening to nearly 3× by $p_{\rm loss}=5e-3$ ($P2 = 0.0218$ vs. $P3 = 0.0629$, $P4 = 0.0714$).

Neutral-atom data alone show a platform where LERs are uniformly higher across the board (roughly 0.0012 at the smallest $p_{\rm loss}$ up to 0.18 at the largest), consistent with its higher intrinsic gate error rate, but the baseline's degradation with $p_{\rm loss}$ is comparatively less explosive than what one sees in the ion data — P0 rises about 17× (0.0101 → 0.177) for a 5× increase in $p_{\rm loss}$, versus the oracle P1's more modest ~8× rise (0.0030 → 0.0244). Because failure counts are larger here even at low $p_{\rm loss}$, the neutral-atom numbers are less shot-noise-limited, making trends across the full $p_{\rm loss}$ range more reliable to read at face value.

The neutral-atom data also show the inference protocols retaining their edge over P2 slightly longer: P3 still leads at $p_{\rm loss}$ = 5e-4 (0.00210 vs. 0.00222), the two are essentially tied by $p_{\rm loss}$ = 1e-3 (0.00429 vs. 0.00419), and only at $p_{\rm loss}$ = 5e-3 does P2 pull decisively ahead (0.0461 vs. 0.0994 for P3). Taken on its own, this dataset suggests that for a noisier physical platform, statistical inference of loss location remains competitive with physical LDU detection over a broader loss-rate window before high-loss ambiguity finally erodes its advantage.

\section{Conclusion}
\label{sec:conclusion}

In this work, we established a criterion for \emph{detectable} qubit loss in stabilizer measurements without leakage-detection units. Our central theoretical observation is that, after puncturing by loss, stabilizer generators that were originally commuting may become effectively anticommuting on the surviving qubits. This loss-induced anticommutation provides an operational signature of detectable loss, because it leads to random measurement outcomes under repeated syndrome extraction.

Based on this condition, we showed how qubit loss can be inferred directly from the set of stabilizer checks exhibiting random outcomes. In this sense, the detectable-loss condition is the foundation that connects the algebraic effect of puncturing to an experimentally accessible syndrome signature.

We then formulated the exact and maximum-likelihood loss inference problems for general stabilizer codes in terms of the observed random-check set. This provides a general framework for loss inference at the level of stabilizer measurements, without requiring dedicated loss-detection hardware.

These results clarify the theoretical role of loss-induced anticommutation in stabilizer measurements and show that it can be turned into an inference procedure, which may result in the reduction of space-time overhead for fault-tolerant quantum computation.

We numerically validated this framework on the rotated surface code under realistic circuit-level noise models for trapped-ion and neutral-atom platforms, comparing inference-based reload against an oracle upper bound, a noisy-LDU reload baseline, and a no-correction baseline. The inference-based protocol matches or exceeds the performance of a physical leakage-detection unit in the low-to-moderate loss-rate regime relevant to near-term hardware, confirming that loss mitigation can be recovered from syndrome data alone without dedicated loss-detection hardware.

An important direction for future work is to extend the present framework beyond idealized random /deterministic classification and combine an inference algorithm with a decoder to optimize the whole correction process. Another direction is to generalize the noise model to account for leakage error and auxiliary loss, which does not satisfy the assumption on the two-qubit gate behavior, such as superconducting qubit systems~\cite{mehta2025bias}. As discussed in Appendix~\ref{app:hook_error}, the same non-entangling assumption also modifies hook-error propagation during syndrome extraction, suggesting that loss-dependent fault propagation may be another important ingredient in future circuit-level decoding strategies. 
The auxiliary loss induces an uncertain syndrome outcome that may disturb our statistical approach.
As future work, it would be interesting to extend this study by applying the distinction of auxiliary loss through SSR.
It is also interesting to specify the inference-friendly code family that has less ambiguity on the randomized syndrome for different loss patterns and Pauli errors, which may result in higher inference accuracy. Another important direction is to move beyond the offline classification framework adopted in this work. The present simulations use an offline classifier based on repeated-shot statistics. Developing a single-shot online inference algorithm suitable for real-time fault-tolerant operation remains an important direction for future work.\\

\section*{Acknowledgment}
SN proposed the concept of detectable loss condition and loss inference, implemented all simulation flows, and prepared the manuscript. TU contributed to the discussion on the formulation of the relaxed loss inference problem as the set cover problem. FK contributed to the discussion on the noise model assumption on neutral atom devices. TS contributed to the discussion on the feasibility of the protocols. DB contributed to the gauge operator style definition on the punctured checks and supervised the project.

SN would like to thank Yasunari Suzuki, Yosuke Ueno, Shuwen Kan, George Umbrarescu, and William John Munro for useful discussions. SN would like to thank Koichiro Miyanishi, Aoi Hayashi, and Takahiro Tsunoda for discussions regarding the assumption underlying noise models for hardware platforms. SN would like to thank members of the ``Software, Algorithms, and (Circuit) Optimization for Efficient Realization of Novel Device Architectures'' group for their invaluable advice regarding the formulation of the loss inference problem. 
This project was supported by the JST Moonshot R\&D Grant Number JPMJMS256G. 
SN acknowledges support from a JSPS Overseas Research Fellowship.
FK acknowledges support from a JSPS KAKENHI Grant Number 25KJ0445.
\bibliographystyle{unsrt}
\bibliography{main.bib}
\appendix
\section{Loss-Induced Changes in Hook-Error Propagation}
\label{app:hook_error}
In the main text, we focused on the non-deterministic syndrome outcomes induced by punctured stabilizer checks. Qubit loss can also modify the propagation of Pauli errors through a syndrome-extraction circuit. In this appendix, we show that a loss occurring before syndrome extraction can change the support and orientation of a hook error, and consequently change whether the resulting correlated data-qubit error reduces the effective code distance. Note that Ref.~\cite{liu2026achieving} also addresses hook errors induced by loss errors. In this appendix, however, we focus on the effect of loss on hook errors arising from Pauli errors on auxiliary qubits, rather than on hook errors directly caused by loss errors.

A hook error is a correlated data-qubit error caused by the propagation of a fault on a syndrome auxiliary qubit through the sequence of CX gates used for syndrome extraction. Depending on the gate ordering, a single fault on an auxiliary qubit can propagate to multiple neighbouring data qubits. Figure~\ref{fig:hook_example} illustrates such a propagation in a rotated surface code. Typically for surface codes, only hook errors propagating to the last two data qubits are considered harmful because errors propagating to three or four data qubits are equivalent, up to the measured stabilizer, to lower-weight errors (one and zero) and therefore do not reduce the effective code distance. Consequently, for each stabilizer type $X$ and $Z$, the last two data qubits in the CX sequence should be chosen so that the resulting weight-two hook error is oriented perpendicular to the minimum-weight logical operator of the corresponding Pauli type, thereby preventing the hook error from reducing the effective code distance.

\begin{figure}[htbp]
    \centering
    \includegraphics[width=\linewidth]{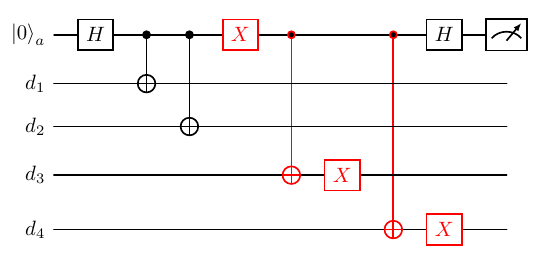}
    \caption{A hook error on a loss-free $XXXX$ syndrome extraction circuit. An $X$ fault on the auxiliary occurring between the second and third CX gates propagates to $d_3$ and $d_4$ through the red-highlighted CX gates, producing a weight-two hook error.}
    \label{fig:hook_example}
\end{figure}

Suppose that the third data qubit is lost before syndrome extraction. Under Assumption~\ref{assumption_noise}, the interaction between the auxiliary and the lost qubit becomes non-entangling and therefore no longer propagates Pauli errors. Consequently, an auxiliary fault occurring at the same time and location in the circuit propagates only through the remaining CX gates, resulting in $X_4$. It therefore reduces the effect of the same hook error pattern.

However, it increases the effect of a different pattern of a hook error. The auxiliary $X$ error before the second CX gate would propagate to the second, third, and fourth data qubits for the case of a loss-free circuit. In the presence of loss error before the second CX gate, now it introduce hook error on $d_2$ and $d_4$ as illustrated in Fig.~\ref{fig:hook_loss}.

\begin{figure}[htbp]
    \centering
    \includegraphics[width=\linewidth]{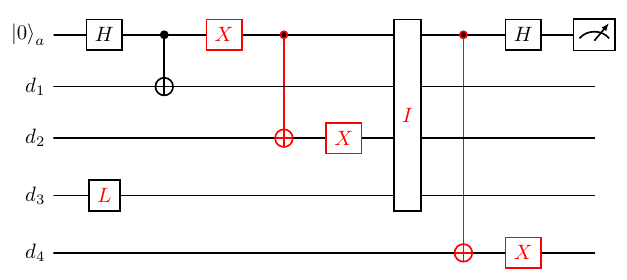}
    \caption{If $d_3$ is lost before syndrome extraction, the corresponding interaction becomes non-entangling. An auxiliary fault occurring before the second CX gate therefore propagates to $d_2$ and $d_4$, instead of propagating to $d_2$, $d_3$, and $d_4$. In the loss-free circuit, $d_2 d_3 d_4$ was equivalent to a single-qubit error on $d_1$ up to stabilizer and is therefore not so harmless. After the loss event, however, the resulting hook error is $d_2 d_4$, which is a weight-two error and may reduce the effective code distance.}
    \label{fig:hook_loss}
\end{figure}

As a consequence, qubit loss can qualitatively change the effective propagation of hook errors. A correlated error that is stabilizer-equivalent to a single-qubit error in the loss-free circuit may instead become a genuine weight-two hook error after a loss event. Conversely, the removal of a propagation path may also suppress a harmful hook error, depending on the loss location.

This observation highlights another way in which qubit loss modifies the effective syndrome-extraction circuit. While the main text focuses on the resulting punctured stabilizer measurements, qubit loss can also alter the propagation of correlated Pauli errors.

The previous example demonstrates that qubit loss can modify hook-error propagation.
We now characterize more generally which lost interactions can affect the propagated support of an auxiliary fault. As we show below, only interactions that would otherwise propagate the fault can modify the resulting hook error, whereas losses occurring before the fault location have no effect.

Consider the syndrome-extraction circuit for an arbitrary stabilizer generator $s=P_1P_2\cdots P_w,$
where the auxiliary qubit sequentially interacts with data qubits $d_1,d_2,\ldots,d_w$
through a prescribed sequence of CX gates.
Suppose that an auxiliary Pauli fault occurs immediately after the $k$-th interaction.
Under the standard propagation rules, the fault propagates only through the remaining interactions,
so that the propagated support is
\begin{equation}
H_k=\{d_{k+1},\ldots,d_w\}.    
\end{equation}

Now suppose that a subset of the data qubits
$L\subseteq\{d_1,\ldots,d_w\}$
is lost before syndrome extraction.
Under Assumption~\ref{assumption_noise}, every interaction involving a lost qubit becomes non-entangling and therefore no longer contributes to fault propagation.
The propagated support is consequently modified to
\begin{equation}
H_k(L)=H_k\setminus L.    
\end{equation}

It follows immediately that the hook-error support is unchanged if and only if
\begin{equation}
L\cap H_k=\varnothing .  
\end{equation}

Equivalently, only the removal of interactions belonging to the propagated suffix $H_k$ can modify the resulting hook error.
Losses occurring before the fault location do not affect its propagation.

A modification of the propagated support does not necessarily change whether a hook error is harmful in general stabilizer codes.

We now specialize the discussion to the rotated surface code, where the harmfulness of a hook error depends not only on its weight but also on its orientation with respect to the minimum-weight logical operators.
We therefore ask under what conditions qubit loss changes this orientation.

Let $\mathcal O(H_k)$
denote the orientation of the propagated support with respect to the corresponding minimum-weight logical operator.
A loss pattern $L$ changes the hook orientation whenever
\begin{equation}
\mathcal O(H_k(L)) \neq \mathcal O(H_k).
\end{equation}
In particular, this requires
$L\cap H_k\neq\varnothing$, since otherwise $H_k(L)=H_k$.

For weight-four stabilizer measurements in the rotated surface code, the propagated suffix consists of the final two interactions. Removing either of these interactions may change the orientation of the resulting weight-two correlated error, whereas removing only earlier interactions leaves the hook orientation unchanged.

Because the numerical simulations in Sec.~\ref{sec:eval} are performed under a circuit-level Pauli noise model, this modification of hook-error propagation is naturally included in the simulated logical error rates. However, its contribution is not isolated from the other effects of qubit loss, such as punctured stabilizer measurements and the resulting loss inference, and a quantitative analysis of its impact is left for future work.

If loss locations can be inferred, the corresponding modification of hook-error propagation could in principle be incorporated into the decoder by adapting its circuit-level fault model. Whether such a loss-conditioned decoder provides a practical advantage remains an interesting direction for future work.

\section{Measurement-Operational Knill--Laflamme Conditions for Loss Inference}
\label{app:operational-kl}
In this appendix, we formulate an operational version of the Knill--Laflamme perspective adapted to loss inference from stabilizer-measurement syndrome data. The purpose is not to replace the standard quantum error-correction condition, but to separate three logically distinct notions: detecting that some loss has occurred, identifying which loss has occurred, and recovering the encoded information after that loss.

Let $\Pi$ denote the projector onto the code space \(\mathcal C\).
For a family of error maps \(\{\mathcal E_a\}_{a\in\mathcal A}\),
the standard Knill--Laflamme condition~\cite{knill1997theory} for correctability can be
written, for Kraus operators \(E_{a,\mu}\), as
\begin{equation}
    \Pi E_{a,\mu}^{\dagger} E_{b,\nu} \Pi
    =
    c_{a\mu,b\nu} \Pi .
\end{equation}
This condition is a statement about the preservation of the encoded
quantum information: after the error, there exists a recovery map that
restores the logical state without learning any logical information.

In the present work, however, the relevant question is different. We
do not assume direct access to a loss flag or to the outcome of a
leakage-detection unit. Instead, after a possible loss event, we
repeatedly measure stabilizer checks and observe a classical syndrome
record. Thus, the first question is not whether the loss is already
correctable, but whether the available measurement procedure produces
classical statistics that distinguish the loss process from the
loss-free process.

We model the repeated syndrome-extraction procedure by a quantum
instrument
\begin{equation}
    \{\mathcal M_y\}_{y\in\mathcal Y},
\end{equation}
where \(y\) denotes the full classical measurement record. For an
error process $\mathcal E_a$, define
\begin{equation}
    p_a(y|\rho)
    :=
    \operatorname{Tr}
    \left[
        \mathcal M_y\!\left(\mathcal E_a(\rho)\right)
    \right],
\end{equation}
for any code state $\rho=\Pi\rho \Pi$.

\begin{definition}[Operational detectability]
\label{def:operational-detectability}
An error process \(\mathcal E_a\) is operationally detectable against
the no-error process \(\mathcal E_I\), with respect to the measurement
instrument $\{\mathcal M_y\}_{y\in\mathcal Y}$, if there exists a
code state $\rho=\Pi\rho \Pi$ such that
\begin{equation}
    p_a(\cdot|\rho) \neq p_I(\cdot|\rho).
\end{equation}
Equivalently, there exists an outcome \(y\in\mathcal Y\) for which
\begin{equation}
    p_a(y|\rho) \neq p_I(y|\rho).
\end{equation}
\end{definition}

\begin{definition}[Operational identifiability]
\label{def:operational-identifiability}
Two error processes \(\mathcal E_a\) and \(\mathcal E_b\) are
operationally identifiable, with respect to
\(\{\mathcal M_y\}_{y\in\mathcal Y}\), if there exists a code state
\(\rho=\Pi\rho \Pi\) such that
\begin{equation}
    p_a(\cdot|\rho) \neq p_b(\cdot|\rho).
\end{equation}
If
\begin{equation}
    p_a(\cdot|\rho) = p_b(\cdot|\rho)
\end{equation}
for all code states \(\rho=\Pi\rho \Pi\), then the two error processes are
observationally equivalent under the chosen measurement procedure.
\end{definition}

\begin{definition}[Operational recoverability after measurement]
\label{def:operational-recoverability}
A family of error processes \(\{\mathcal E_a\}_{a\in\mathcal A}\) is
operationally recoverable with respect to the measurement instrument
\(\{\mathcal M_y\}_{y\in\mathcal Y}\) if there exists a collection of
recovery maps \(\{\mathcal R_y\}_{y\in\mathcal Y}\), depending only on
the observed classical outcome \(y\), such that for every
\(a\in\mathcal A\) and every code state \(\rho=P\rho P\),
\begin{equation}
    \sum_{y\in\mathcal Y}
    \mathcal R_y
    \left(
        \mathcal M_y\!\left(\mathcal E_a(\rho)\right)
    \right)
    =
    \rho,
\end{equation}
up to the intended logical identity channel.
\end{definition}

These definitions give the hierarchy relevant to this work.
Recoverability is the strongest requirement, because it asks whether
the encoded information can be restored. Identifiability asks whether
distinct physical loss patterns can be distinguished from the
classical measurement record. Detectability is weaker: it asks only
whether a loss event can be distinguished from the loss-free case.

The following proposition connects the algebraic condition used in the
main text to the operational notion of detectability defined above.

\begin{proposition}[Anticommuting punctured checks imply operational detectability]
\label{prop:anticommuting-punctured-detectability}
Let $\Pi$ be the projector onto a stabilizer code space, and let $L$
be a loss pattern. Suppose that, under the post-loss
syndrome-extraction instrument, two stabilizer generators $s_i$ and
\(s_j\) are effectively measured as punctured Pauli operators
\(s_i'\) and \(s_j'\). Assume that the repeated syndrome-extraction
record contains measurements of both effective checks under the same
post-loss effective-check model. If
\begin{equation}
    s_i' s_j' = - s_j' s_i',
\end{equation}
then the loss process associated with \(L\) is operationally
detectable against the loss-free process with respect to repeated
stabilizer measurement.
\end{proposition}

\begin{proof}
In the loss-free case, every code state \(\rho=\Pi\rho \Pi\) is stabilized
by the original stabilizer generators. Therefore, the corresponding
stabilizer-measurement outcomes are deterministic.

After the loss pattern $L$, the relevant effective observables are
the punctured operators $s_i'$ and $s_j'$. Since these operators
anticommute, no quantum state can be a simultaneous eigenstate of both
of them. Hence their repeated measurement outcomes cannot both remain
deterministic in the same way as in the loss-free stabilizer
measurement.

Therefore, the classical outcome distribution produced by the
post-loss repeated syndrome-extraction instrument differs from the
loss-free distribution for code states. By
Definition~\ref{def:operational-detectability}, the loss process is
operationally detectable.
\end{proof}

Combining this proposition with the parity criterion in the main text
gives the following operational form of the detectability condition.

\begin{corollary}[Parity criterion for operational detectability]
\label{cor:parity-operational-detectability}
Let \(S=\langle s_1,\ldots,s_m\rangle\) be a stabilizer group, and let
\(L\) be a loss pattern. For two stabilizer generators \(s_i\) and
\(s_j\), define
\begin{equation}
    A_{ij}
    :=
    \{q \mid (s_i)_q (s_j)_q = - (s_j)_q (s_i)_q\}.
\end{equation}
If there exists a pair \(i\neq j\) such that
\begin{equation}
    |L\cap A_{ij}|
\end{equation}
is odd, then the loss process associated with \(L\) is operationally
detectable with respect to repeated stabilizer measurement.
\end{corollary}

\begin{proof}
By the parity criterion in the main text, the oddness of
\(|L\cap A_{ij}|\) implies that the punctured operators \(s_i'\) and
\(s_j'\) anticommute. The claim then follows immediately from
Proposition~\ref{prop:anticommuting-punctured-detectability}.
\end{proof}

This operational detectability result does not imply that the loss
pattern \(L\) is uniquely identifiable. Two distinct loss patterns may
induce the same set of non-deterministic checks and hence the same
classical signature under the chosen measurement procedure. This is
why the main text separately formulates the loss-inference problem: detectability gives a witness that some loss has affected the
stabilizer statistics, whereas identifiability asks whether the
observed random-check set determines the underlying loss pattern.

Finally, operational detectability and identifiability should not be
confused with full quantum recoverability. Once a loss location is
identified, it can often be treated as an erasure and incorporated
into a recovery procedure. However, the ability to detect or infer the
loss from syndrome statistics is only the information-gathering part
of the correction process. Full recoverability additionally depends on
the code distance, the number and geometry of losses, and the
performance of the subsequent erasure and Pauli decoder. 

\section{M\o lmer--S\o rensen Gate on Trapped Ions}
\label{app:MS}

In this appendix, we explain how the non-entangling behavior in Assumption~\ref{assumption_noise} arises for an ideal M\o lmer--S\o rensen (MS) interaction~\cite{molmer1999multiparticle}. We first review the ideal two-ion single-mode model and then discuss how ion loss may modify the motional-mode structure and produce an implementation-dependent local channel on the surviving ion.

Consider first the intact two-ion crystal before any loss event.
The two trapped-ion qubits, labelled $L$ and $S$, are encoded in long-lived internal states $\{\ket{0}_j,\ket{1}_j\}$ and coupled to a selected normal mode of the two-ion crystal. Let $\nu_{2}$ denote the frequency of this mode, and let $a$ and $a^\dagger$ be its annihilation and creation operators, with $[a,a^\dagger]=1$. Within the Lamb--Dicke and rotating-wave approximations, and after retaining only this selected mode, a bichromatic field whose beat-note frequencies are symmetrically detuned by $\delta$ from the corresponding red and blue sidebands gives the effective interaction-picture Hamiltonian
\begin{equation}
H_I(t)
 =
 \hbar
 \left(
     a e^{-i\delta t}
     +
     a^\dagger e^{i\delta t}
 \right)
 \sum_{j\in\{L,S\}}
 g_j \sigma_\phi^{(j)},
\label{eq:MS_Hamiltonian}
\end{equation}
where $g_j=\eta_j\Omega_j$ is the spin--motion coupling strength and
$\sigma_\phi^{(j)}
 =\cos\phi\,\sigma_x^{(j)}
 +\sin\phi\,\sigma_y^{(j)}$.
Defining
\begin{equation}
F_\phi
 =
 \sum_{j\in\{L,S\}}g_j\sigma_\phi^{(j)},
\end{equation}
the time evolution by Eq.~\eqref{eq:MS_Hamiltonian} can be written as
\begin{equation}
U(t) = D\!\left[\alpha(t)F_\phi\right] \exp\!\left[i\Phi(t)F_\phi^2\right],
\label{eq:MS_evolution}
\end{equation}
where $D(\alpha)=\exp(\alpha a^\dagger-\alpha^*a)$ is the displacement operator acting on the shared motional mode, and
\begin{align}
\alpha(t)
 &=
 \frac{1-e^{i\delta t}}{\delta},\\
\Phi(t)
 &=
 \frac{\delta t-\sin(\delta t)}{\delta^2},
\end{align}
up to convention-dependent signs and coupling factors.

At the gate time $\tau=2\pi/\delta$, the displacement vanishes due to $\alpha(\tau)=0$, and the spin and motional degrees of freedom are disentangled. Expanding the remaining spin operators gives
\begin{equation}
F_\phi^2
 =
 \left(g_L^2+g_S^2\right)I
 +
 2g_Lg_S
 \sigma_\phi^{(L)}\sigma_\phi^{(S)}.
\label{eq:MS_force_square}
\end{equation}
The first term produces only a global phase, whereas the cross term generates the entangling MS interaction.

Suppose now that ion $L$ is lost before the laser pulse is applied.
Within the idealized model in which the remaining motional parameters are
kept fixed, the loss is represented by setting $g_L=0$, and hence
\begin{equation}
F_\phi
 =
 g_S\sigma_\phi^{(S)}.
\end{equation}
The evolution becomes
\begin{equation}
U_{\mathrm{loss}}(t)
 =
 D\!\left[
     \alpha(t)g_S\sigma_\phi^{(S)}
   \right]
 \exp\!\left[
     -i\Phi(t)g_S^2
     \left(\sigma_\phi^{(S)}\right)^2
   \right].
\end{equation}
Since
\begin{equation}
\left(\sigma_\phi^{(S)}\right)^2=I,
\end{equation}
the second factor is a global phase and cannot entangle the surviving ion with any other qubit. 
    If, as an additional idealization, the post-loss mode has the same effective detuning as the original two-ion mode, $\delta_{\rm loss}=\delta$, then the nominal gate duration $\tau=2\pi/|\delta|$ also closes the post-loss phase-space trajectory:
\begin{equation}
\alpha_1(\tau)=0.
\end{equation}
Under this additional closure condition,
\begin{equation}
U_{\mathrm{loss}}(\tau)= e^{-i\Phi(\tau)g_S^2} I_S,
\end{equation}
up to a global phase.

In general, however, ion loss may change the mode frequency and hence the effective detuning, so that $\delta_{\rm loss}\neq\delta$ and $\alpha_1(\tau)\neq 0$. In that case, the nominal two-ion pulse can leave a residual spin-dependent displacement involving only the surviving ion and the post-loss motional mode.


\begin{proposition}
Within the ideal single-mode MS model of
Eq.~\eqref{eq:MS_Hamiltonian}, loss of one ion before the gate removes
the two-ion interaction. If the post-loss phase-space trajectory closes
at the nominal gate time $\tau$, the surviving ion undergoes no
nontrivial operation, up to a global phase.
\end{proposition}

The conclusion above applies to the ideal spin-dependent-force model. A physical implementation may contain additional terms that are not included in Eq.~\eqref{eq:MS_Hamiltonian}, such as off-resonant carrier coupling, differential AC Stark shifts, coupling to spectator modes, pulse-shape imperfections, or errors in the phase-space closure condition. After one ion is lost, these terms may leave a residual operation on the surviving ion. Since the lost ion is absent, however, such an operation remains local to the surviving subsystem and cannot restore the intended two-ion entangling interaction.

For a fixed realization, the residual operation is generally a single-qubit channel and need not be a Pauli channel. The Pauli-channel
model in Assumption~\ref{assumption_noise} should therefore be understood as an effective stochastic description obtained, for example, after averaging over fluctuating control parameters or after Pauli-frame randomization. Under this effective description,
\begin{equation}
\label{eq:single-qubit_Pauli_error}
\mathcal E_S(\rho)
 =
 \sum_{P\in\{I,X,Y,Z\}}
 p_P P\rho P,
\end{equation}
which is precisely the non-entangling behavior assumed in the main text.

Thus, the ideal MS interaction provides a direct realization of the non-entangling part of Assumption~\ref{assumption_noise}. When one ion
is lost before the gate, the entangling cross term vanishes. The additional Pauli channel represents implementation-dependent local
faults on the surviving ion rather than an intrinsic consequence of the ideal MS Hamiltonian.

\section{Loss treatments on Rydberg atoms}
This appendix has two purposes. First, we justify Assumption~\ref{assumption_noise} for the ideal Rydberg-blockade CZ gate. Second, we briefly review state-selective readout (SSR),
which serves as a representative hardware-based loss-detection approach discussed in Sec.~\ref{sec:preliminary}.

\subsection{CZ gate via Rydberg blockade}
\label{app:blockade-CZ}
Let us first consider a CZ gate of Rydberg blockade on a pair of neutral atoms.
This type of atom qubit is often called a Rydberg qubit.
Due to the many sublevels inside an atom, there is a variety of qubit encodings.
Encoding a qubit into the hyperfine structure of the ground states or metastable states is the typical encoding for a Rydberg qubit~\cite{PhysRevX.12.021027, PRXQuantum.6.020334, senoo2025, Ma2023, Zhang_2026}.
Therefore, the properties of Rydberg qubits are not uniformly understood because they entirely depend on the atomic species and the encoding type of qubit.
However, their gate operation can be explained in a common framework since every encoded qubit has common features in the two-qubit gate operation.
Thereafter, we assume a simplified structure of a Rydberg qubit which consists of four levels: qubit states $\ket{0}$ and $\ket{1}$, Rydberg state $\ket{r}$, and the leakage state $\ket{L}$.

A typical state manipulation on Rydberg qubits is generally described as a two-level system consisting of $\ket{\xi}$ and $\eta$ operated by a monochromatic optical field with the local phase $\varphi(t)$ and the coupling frequency~$\omega$. This dynamics is described as the Hamiltonian
\begin{equation}
  H^{\xi \eta} = \left( \frac{\Omega(t)}{2} e^{i\varphi(t)}\ket{\xi}\bra{\eta} + \text{h.c.} \right) -\Delta(t) \ket{\eta}\bra{\eta}, \label{eq:RWA_Hamiltonian}
\end{equation}
where $\Omega(t)$ and $\Delta(t)=\omega-\omega_0$ are the Rabi frequency and the local detuning from the resonance frequency of $\omega_0$, respectively.

The origin of interactions between Rydberg qubits is the dipole-dipole
interaction between Rydberg states~\cite{PhysRevLett.85.2208}. 
Consider the dipole-dipole potential between the atom labelled by c and t denoting ``control'' and ``target'', respectively. Then suppose $\ket{1}_j$ state of both Rydberg qubits are coupled with the same Rydberg state $\ket{r}_j$, where $j$ is the label of two atoms either c or t.
The interaction dynamics of a two-qubit system through the Rydberg state is described as the Hamiltonian~\cite{PhysRevLett.85.2208, morgado2021quantum}
\begin{equation}\label{eq:Rydberg_Hamiltonian}
\begin{split}
  &H(t) = V_{c,t}\ket{rr}_{c,t}\bra{rr}\\
  &+ \sum_{j\in\{c,t\}} \left[\left( \Delta_j(t)-i\gamma\right)\ket{r}_j\bra{r}-\frac{\Omega_j(t)}{2}\left(\ket{1}_j\bra{r}_j + \text{h.c.} \right)\right]\\
  &=H^{rrrr}_{c,t} + \sum_{j\in\{c,t\}}H^{1r}_{j},
\end{split}
\end{equation}
where $V_{c,t}$, $\delta_j(t)$, $\gamma$ and $\Omega_j(t)$ are the energy shift inducing Rydberg blockade, the detuning of the laser field from Rydberg state, decay rate from $\ket{r}$, and Rabi frequency of the laser field to excite $\ket{0}_j$ to $\ket{r}_j$.

The main noise source of loss is the decay from $\ket{r}$ denoted by $\gamma$ in the Eq.~\eqref{eq:Rydberg_Hamiltonian}.
This decay can be assumed as a transition to $\ket{L}$ where the transited state has never return to the qubit space effectively.
This error is a dominant noise source for the neutral atom quantum computation~\cite{wu2022erasure, scholl2023erasure, bluvstein2026fault, jandura2026, fumiyoshi_forthcomin2026}.

The Hamiltonian of Eq.~\eqref{eq:Rydberg_Hamiltonian} can be separated to the interaction term of Rydberg states $H^{rrrr}_{c,t}$ and single-atom excitations $H^{0r}_{j}$
as the right-hand side of the equation.
Then the CZ gate is realised the sequential operations
\begin{equation}\label{eq:Rydberg_blockade_gate}
\begin{split}
   U_{\text{CZ}} =& \exp[-iH^{1r}_c\tau_1]\exp[-i(H^{r1}_c+H^{rrrr}_{c,t})\tau_2]\\
   &\times \exp[-i(H^{1r}_1)\tau_1]
\end{split}
\end{equation}
with $\Delta_c=\Delta_t=0$, $\Omega_c=\Omega_t=\Omega$, $\tau_1=\pi/\Omega$ and $\tau_2=2\pi/\Omega$. This sequence represents applying $\pi$-pulse for the transition $\ket{1}_c\rightarrow\ket{r}_c$ at first, then applying 2$\pi$-pulse for $\ket{1}_t\rightarrow\ket{r}_t$ and finally applying $\pi$-pulse again to return back$\ket{r}_c\rightarrow\ket{1}_c$.
The strong Rydberg blockade prevents the second Rydberg excitation of the target qubit when $|V_{c,t}|\gg \Omega$.
While ideally the first and final excitations cause no effect to the control qubit, each
computational state returns to itself after the entire sequence but acquires a $\pi$ phase shift during each Rabi oscillation due to the Rydberg blockade.

In this implementation, the entangling phase arises only when both atoms are present and interact through the Rydberg blockade mechanism.
If one of the atoms is lost, the blockade interaction is absent, and the intended two-qubit gate no longer generates entanglement.
As a result, when qubit $L$ is lost and qubit $S$ survives, the action of the gate ideally attributes to the identity operation on the surviving qubit $S$:
$$
\ket{\phi}_S \longmapsto \ket{\phi}_S .
$$
This is the reason why the assumption in the main text is satisfied in the special case where no additional Pauli error is induced on the surviving qubit, namely $p_I=1$ and $p_X = p_Y = p_Z = 0$.

As described in Appendix~\ref{app:MS}, additional noise terms originating from the experimental setup affect the qubit as noise sources during the gate operation. 
Even though there are these extra noise sources, the noise on the surviving qubit can be assumed to act locally as the same as discussed in Appendix~\ref{app:MS} as long as the surviving atom qubits are isolated from each other with enough distance. Under this condition, the noise on the surviving qubit is assumed to be Eq.~\eqref{eq:single-qubit_Pauli_error}. Thus, Assumption~\ref{assumption_noise} is also satisfied by the CZ gate for Rydberg atoms.

\subsection{State-selective readout for auxiliary qubit loss}
\label{app:atom-SSR}
For neutral atom qubits, the qubit state encoded in the hyperfine ground or metastable manifolds is measured through resonance fluorescence state-selectively, called state-selective readout (SSR)~\cite{PhysRevA.108.03240, PRXQuantum.4.030337}.
One qubit state, for example $\ket{0}$, is coupled to a cycling transition and scatters many photons as the ``bright'' state, while the other, for example $\ket{1}$, remains decoupled from the probe light as the ``dark'' state.
The qubit state is then inferred from the detection of each selected state. 

For example of the ground-state qubit of $^{1}\text{S}_0$ with hyperfine structure $F=1/2$ on $^{171}\text{Yb}$ in the figure~\ref{fig:state-selective readout}, the qubit is defined as $\ket{m_F = -1/2}\equiv \ket{0}$ and $\ket{m_F = -1/2}\equiv \ket{0}$. Assuming that the narrow transition $^{1}\text{S}_0\leftrightarrow ^{3}\text{P}_1 (F=3/2)$ especially targeting the $m_F=-3/2$ of $^{3}\text{P}_1$, $\ket{0}$ is a decayed state only allowed from $^{3}\text{P}_1$ $m_F=-3/2$ due to the selection rule.
As a result, occupation of $\ket{0}$ is detectable as scattering of many photons (bright state), while $\ket{1}$ is entirely off-resonant and dark.
These bright and dark fluorescences are directly mapped onto the occupation of $\ket{0}$.
Although the bright result indicates $\ket{0}$, the dark result is ambiguous: a dark result is consistent with the atom being in $\ket{1}$ or the loss.
To resolve this, for example, consider targeting $m_F=3/2$ of $^{3}\text{P}_1$, which is the symmetrically opposite transition. Since the fluorescence of this transition can be mapped onto the occupation of $\ket{1}$ oppositely, sequential readouts of $\ket{0}$ and $\ket{1}$ allow us to distinguish $\ket{0}, \ket{1}$ and loss.

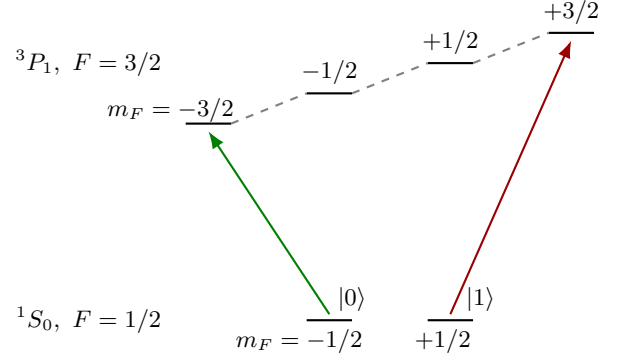
\begin{figure}[htbp]
    \centering
    \hspace{5pt}
    \begin{tikzpicture}[>=Latex, thick, scale=1.0]
      \draw (0.0,2.6) -- (0.6,2.6);
      \draw (1.6,3.0) -- (2.2,3.0);
      \draw (3.2,3.4) -- (3.8,3.4);
      \draw (4.8,3.8) -- (5.4,3.8);
      \draw[dashed, gray] (0.6,2.6) -- (1.6,3.0);
      \draw[dashed, gray] (2.2,3.0) -- (3.2,3.4);
      \draw[dashed, gray] (3.8,3.4) -- (4.8,3.8);
    
      \node[above] at (-0.2,2.5)  {\small $m_F=-3/2$};
      \node[above] at (1.9,3.0)  {\small $-1/2$};
      \node[above] at (3.5,3.4)  {\small $+1/2$};
      \node[above] at (5.1,3.8)  {\small $+3/2$};
      \node[left] at (-0.2,3.4) {$^3P_1,\ F=3/2$};
    
      \draw (1.6,0) -- (2.2,0);
      \draw (3.2,0) -- (3.8,0);
      \node[below] at (1.50,0) {$m_F=-1/2$};
      \node[above] at (2.20,0) {$|0\rangle$};
      \node[below] at (3.4,0) {$+1/2$};
      \node[above] at (3.9,0) {$|1\rangle$};
      \node[left] at (-0.2,0) {$^1S_0,\ F=1/2$};
    
      \draw[->, thick, green!50!black] (1.9,0.08) -- (0.3,2.5);
      \draw[->, thick, red!60!black]   (3.5,0.08) -- (5.1,3.7);
    \end{tikzpicture}
    \caption{A schematic example of the SSR of the ground-state qubit on $^{171}Yb$. The green and red arrows represents the transitions of $m_F=-3/2$ of $^{3}\text{P}_1$ and $m_F=3/2$ of $^{3}\text{P}_1$. Driving these transitions sequentially, the occupation of quantum states in $\{\ket{0}, \ket{1}\}$ can be known separately. Mapping both measurement readouts into the qubit state, it is possible to distinguish $\ket{0}$ (bright, dark), $\ket{1}$ (dark, bright) and loss (dark, dark).}
    \label{fig:state-selective readout}
\end{figure}

The overhead of this readout scheme is comparatively smaller than introducing LRU since it only requires few additional probe lasers, while the LRU requires frequent CZ gates or measurement of auxiliary qubits.
Thanks to the capability of detecting loss of a measured atom, this readout provides high detection fidelity and distinguishes the measured states to be $\ket{0}$, $\ket{1}$ or loss.
This atom loss information makes the qubit loss detectable. Moreover, it enables decoders to recognise the places where an auxiliary qubit is lost, which is expected to be improved~\cite{baranes2026leveraging}.
\end{document}

%% file: preamble.tex
\usepackage[english]{babel}
\usepackage[utf8]{inputenc}
\usepackage[colorinlistoftodos, color=green!40, prependcaption]{todonotes}
\usepackage[pdftex, pdftitle={Article}, pdfauthor={Author}]{hyperref} 
\usepackage{mathrsfs}
\usetikzlibrary{arrows.meta}
\usepackage{amsthm}
\usepackage{tcolorbox}
\usepackage[caption=false]{subfig}
\newtheorem{problem}{Problem}

\usepackage{amsthm}
\usepackage{mathtools}
\usepackage{physics}
\usepackage{xcolor}
\usepackage{graphicx}
\usepackage{algorithm}
\usepackage{algpseudocode}
\usepackage[left=23mm,right=13mm,top=35mm,columnsep=15pt]{geometry} 
\usepackage{adjustbox}
\usepackage{placeins}
\usepackage{csquotes}
\usepackage{wrapfig}
\usepackage[english]{babel}
\usepackage[utf8]{inputenc}
\usepackage{braket}
\usepackage{mathrsfs}
\usepackage{amsmath}               
\theoremstyle{plain}
\newtheorem{assumption}{Assumption}
\newtheorem{theorem}{Theorem}
\newtheorem{lemma}{Lemma}
\newtheorem{definition}[theorem]{Definition}
\newtheorem{corollary}[theorem]{Corollary}
\newtheorem*{remark}{Remark}
\newtheorem{proposition}{Proposition}

%% file: main.bib
@article{chow2024circuit,
  title={Circuit-based leakage-to-erasure conversion in a neutral-atom quantum processor},
  author={Chow, Matthew NH and Buchemmavari, Vikas and Omanakuttan, Sivaprasad and Little, Bethany J and Pandey, Saurabh and Deutsch, Ivan H and Jau, Yuan-Yu},
  journal={PRX Quantum},
  volume={5},
  number={4},
  pages={040343},
  year={2024},
  publisher={APS}
}

@article{delfosse2020linear,
  title={Linear-time maximum likelihood decoding of surface codes over the quantum erasure channel},
  author={Delfosse, Nicolas and Z{\'e}mor, Gilles},
  journal={Physical Review Research},
  volume={2},
  number={3},
  pages={033042},
  year={2020},
  publisher={APS}
}

@article{kang2023quantum,
  title={Quantum error correction with metastable states of trapped ions using erasure conversion},
  author={Kang, Mingyu and Campbell, Wesley C and Brown, Kenneth R},
  journal={PRX Quantum},
  volume={4},
  number={2},
  pages={020358},
  year={2023},
  publisher={APS}
}

@article{wu2022erasure,
  title={Erasure conversion for fault-tolerant quantum computing in alkaline earth Rydberg atom arrays},
  author={Wu, Yue and Kolkowitz, Shimon and Puri, Shruti and Thompson, Jeff D},
  journal={Nature communications},
  volume={13},
  number={1},
  pages={4657},
  year={2022},
  publisher={Nature Publishing Group UK London}
}

@article{Fu2025errorcorrectionin,
  doi = {10.22331/q-2025-10-20-1886},
  url = {https://doi.org/10.22331/q-2025-10-20-1886},
  title = {Error {C}orrection in {D}ynamical {C}odes},
  author = {Fu, Esther Xiaozhen and Gottesman, Daniel},
  journal = {{Quantum}},
  issn = {2521-327X},
  publisher = {{Verein zur F{\"{o}}rderung des Open Access Publizierens in den Quantenwissenschaften}},
  volume = {9},
  pages = {1886},
  month = oct,
  year = {2025}
}

@article{knill1997theory,
  title={Theory of quantum error-correcting codes},
  author={Knill, Emanuel and Laflamme, Raymond},
  journal={Physical Review A},
  volume={55},
  number={2},
  pages={900},
  year={1997},
  publisher={APS}
}

@article{bombin2007optimal,
  title={Optimal resources for topological two-dimensional stabilizer codes: Comparative study},
  author={Bomb{\'\i}n, H{\'e}ctor and Martin-Delgado, Miguel A},
  journal={Physical Review A—Atomic, Molecular, and Optical Physics},
  volume={76},
  number={1},
  pages={012305},
  year={2007},
  publisher={APS}
}

@article{baranes2026leveraging,
  title={Leveraging Qubit Loss Detection in Fault-Tolerant Quantum Algorithms},
  author={Baranes, Gefen and Cain, Madelyn and Ataides, J Pablo Bonilla and Bluvstein, Dolev and Sinclair, Josiah and Vuleti{\'c}, Vladan and Zhou, Hengyun and Lukin, Mikhail D},
  journal={Physical Review X},
  volume={16},
  number={1},
  pages={011002},
  year={2026},
  publisher={APS}
}

@article{perrin2025quantum,
  title={Quantum error correction resilient against atom loss},
  author={Perrin, Hugo and Jandura, Sven and Pupillo, Guido},
  journal={Quantum},
  volume={9},
  pages={1884},
  year={2025},
  publisher={Verein zur F{\"o}rderung des Open Access Publizierens in den Quantenwissenschaften}
}

@article{wood2018quantification,
  title={Quantification and characterization of leakage errors},
  author={Wood, Christopher J and Gambetta, Jay M},
  journal={Physical Review A},
  volume={97},
  number={3},
  pages={032306},
  year={2018},
  publisher={APS}
}

@article{stricker2020experimental,
  title={Experimental deterministic correction of qubit loss},
  author={Stricker, Roman and Vodola, Davide and Erhard, Alexander and Postler, Lukas and Meth, Michael and Ringbauer, Martin and Schindler, Philipp and Monz, Thomas and M{\"u}ller, Markus and Blatt, Rainer},
  journal={Nature},
  volume={585},
  number={7824},
  pages={207--210},
  year={2020},
  publisher={Nature Publishing Group UK London}
}

@article{grassl1997codes,
  title={Codes for the quantum erasure channel},
  author={Grassl, Markus and Beth, Th and Pellizzari, Thomas},
  journal={Physical Review A},
  volume={56},
  number={1},
  pages={33},
  year={1997},
  publisher={APS}
}

@article{leahy2019quantum,
  title={Quantum insertion-deletion channels},
  author={Leahy, Janet and Touchette, Dave and Yao, Penghui},
  journal={arXiv preprint arXiv:1901.00984},
  year={2019}
}

@article{aliferis2005fault,
  title={Fault-tolerant quantum computation for local leakage faults},
  author={Aliferis, Panos and Terhal, Barbara M},
  journal={arXiv preprint quant-ph/0511065},
  year={2005}
}

@article{liu2026achieving,
  title={Achieving Optimal-Distance Atom-Loss Correction via Pauli Envelope},
  author={Liu, Pengyu and Tan, Shi Jie Samuel and Huang, Eric and Acar, Umut A and Zhou, Hengyun and Zhao, Chen},
  journal={arXiv preprint arXiv:2603.04156},
  year={2026}
}

@article{wasilewski2007protecting,
  title={Protecting an optical qubit against photon loss},
  author={Wasilewski, Wojciech and Banaszek, Konrad},
  journal={Physical Review A—Atomic, Molecular, and Optical Physics},
  volume={75},
  number={4},
  pages={042316},
  year={2007},
  publisher={APS}
}

@article{kobayashi2025erasure,
  title={Erasure-tolerance scheme for the surface codes on neutral atom quantum computers},
  author={Kobayashi, Fumiyoshi and Nagayama, Shota},
  journal={IEEE Transactions on Quantum Engineering},
  year={2025},
  publisher={IEEE}
}

@article{vala2005quantum,
  title={Quantum error correction of a qubit loss in an addressable atomic system},
  author={Vala, Jiri and Whaley, K Birgitta and Weiss, David S},
  journal={Physical Review A—Atomic, Molecular, and Optical Physics},
  volume={72},
  number={5},
  pages={052318},
  year={2005},
  publisher={APS}
}

@article{hagiwara2025introduction,
  title={Introduction to quantum deletion error-correcting codes},
  author={Hagiwara, Manabu},
  journal={IEICE Transactions on Fundamentals of Electronics, Communications and Computer Sciences},
  volume={108},
  number={3},
  pages={363--375},
  year={2025},
  publisher={The Institute of Electronics, Information and Communication Engineers}
}

@inproceedings{nakayama2020single,
  title={Single quantum deletion error-correcting codes},
  author={Nakayama, Ayumu and Hagiwara, Manabu},
  booktitle={2020 International Symposium on Information Theory and Its Applications (ISITA)},
  pages={329--333},
  year={2020},
  organization={IEEE}
}

@book{gottesman1997stabilizer,
  title={Stabilizer codes and quantum error correction},
  author={Gottesman, Daniel},
  year={1997},
  publisher={California Institute of Technology}
}

@article{de2025mid,
  title={A mid-circuit erasure check on a dual-rail cavity qubit using the joint-photon number-splitting regime of circuit QED},
  author={de Graaf, Stijn J and Xue, Sophia H and Chapman, Benjamin J and Teoh, James D and Tsunoda, Takahiro and Winkel, Patrick and Garmon, John WO and Chang, Kathleen M and Frunzio, Luigi and Puri, Shruti and others},
  journal={npj Quantum Information},
  volume={11},
  number={1},
  pages={1},
  year={2025},
  publisher={Nature Publishing Group UK London}
}

@article{li2025fast,
  title={Fast, continuous and coherent atom replacement in a neutral atom qubit array},
  author={Li, Yiyi and Bao, Yicheng and Peper, Michael and Li, Chenyuan and Thompson, Jeff D},
  journal={arXiv preprint arXiv:2506.15633},
  year={2025}
}

@inproceedings{litteken2022reducing,
  title={Reducing runtime overhead via use-based migration in neutral atom quantum architectures},
  author={Litteken, Andrew and Baker, Jonathan M and Chong, Frederic T},
  booktitle={2022 IEEE International Conference on Quantum Computing and Engineering (QCE)},
  pages={566--576},
  year={2022},
  organization={IEEE}
}

@incollection{karp2009reducibility,
  title={Reducibility among combinatorial problems},
  author={Karp, Richard M},
  booktitle={50 Years of Integer Programming 1958-2008: from the Early Years to the State-of-the-Art},
  pages={219--241},
  year={2009},
  publisher={Springer}
}

@inproceedings{slavik1996tight,
  title={A tight analysis of the greedy algorithm for set cover},
  author={Slav{\'\i}k, Petr},
  booktitle={Proceedings of the twenty-eighth annual ACM symposium on Theory of computing},
  pages={435--441},
  year={1996}
}

@article{delfosse2021almost,
  title={Almost-linear time decoding algorithm for topological codes},
  author={Delfosse, Nicolas and Nickerson, Naomi H},
  journal={Quantum},
  volume={5},
  pages={595},
  year={2021},
  publisher={Verein zur F{\"o}rderung des Open Access Publizierens in den Quantenwissenschaften}
}

@article{barrett2010fault,
  title={Fault tolerant quantum computation with very high threshold for loss errors},
  author={Barrett, Sean D and Stace, Thomas M},
  journal={Physical review letters},
  volume={105},
  number={20},
  pages={200502},
  year={2010},
  publisher={APS}
}

@article{stace2009thresholds,
  title={Thresholds for topological codes in the presence of loss},
  author={Stace, Thomas M and Barrett, Sean D and Doherty, Andrew C},
  journal={Physical review letters},
  volume={102},
  number={20},
  pages={200501},
  year={2009},
  publisher={APS}
}

@article{nishio2025multiplexed,
  title={Multiplexed quantum communication with surface and hypergraph product codes},
  author={Nishio, Shin and Connolly, Nicholas and Piparo, Nicol{\`o} Lo and Munro, William John and Scruby, Thomas Rowan and Nemoto, Kae},
  journal={Quantum},
  volume={9},
  pages={1613},
  year={2025},
  publisher={Verein zur F{\"o}rderung des Open Access Publizierens in den Quantenwissenschaften}
}

@article{google2025quantum,
  title={Quantum error correction below the surface code threshold},
  journal={Nature},
  volume={638},
  number={8052},
  pages={920--926},
  year={2025},
  publisher={Nature Publishing Group UK London}
}

@article{bluvstein2026fault,
  title={A fault-tolerant neutral-atom architecture for universal quantum computation},
  author={Bluvstein, Dolev and Geim, Alexandra A and Li, Sophie H and Evered, Simon J and Bonilla Ataides, J Pablo and Baranes, Gefen and Gu, Andi and Manovitz, Tom and Xu, Muqing and Kalinowski, Marcin and others},
  journal={Nature},
  volume={649},
  number={8095},
  pages={39--46},
  year={2026},
  publisher={Nature Publishing Group UK London}
}

@article{gu2025optimizing,
  title={Optimizing quantum error-correction protocols with erasure qubits},
  author={Gu, Shouzhen and Vaknin, Yotam and Retzker, Alex and Kubica, Aleksander},
  journal={PRX Quantum},
  volume={6},
  number={4},
  pages={040354},
  year={2025},
  publisher={APS}
}

@article{chang2025surface,
  title={Surface code with imperfect erasure checks},
  author={Chang, Kathleen and Singh, Shraddha and Claes, Jahan and Sahay, Kaavya and Teoh, James and Puri, Shruti},
  journal={PRX Quantum},
  volume={6},
  number={4},
  pages={040355},
  year={2025},
  publisher={APS}
}

@article{molmer1999multiparticle,
  title={Multiparticle entanglement of hot trapped ions},
  author={M{\o}lmer, Klaus and S{\o}rensen, Anders},
  journal={Physical Review Letters},
  volume={82},
  number={9},
  pages={1835},
  year={1999},
  publisher={APS}
}

@article{smith2025single,
  title={Single-qubit gates with errors at the 10-7 level},
  author={Smith, Molly C and Leu, Aaron D and Miyanishi, Koichiro and Gely, Mario F and Lucas, David M},
  journal={Physical Review Letters},
  volume={134},
  number={23},
  pages={230601},
  year={2025},
  publisher={APS}
}

@article{hughes2025trapped,
  title={Trapped-ion two-qubit gates with> 99.99\% fidelity without ground-state cooling},
  author={Hughes, AC and Srinivas, R and L{\"o}schnauer, CM and Knaack, HM and Matt, R and Ballance, CJ and Malinowski, M and Harty, TP and Sutherland, RT},
  journal={arXiv preprint arXiv:2510.17286},
  year={2025}
}

@article{sotirova2024high,
  title={High-fidelity heralded quantum state preparation and measurement},
  author={Sotirova, AS and Leppard, JD and Vazquez-Brennan, A and Decoppet, SM and Pokorny, F and Malinowski, M and Ballance, CJ},
  journal={arXiv preprint arXiv:2409.05805},
  year={2024}
}

@article{kiefer2026protected,
  title={Protected quantum gates using qubit doublons in dynamical optical lattices},
  author={Kiefer, Yann and Zhu, Zijie and Fischer, Lars and Jele, Samuel and G{\"a}chter, Marius and Bisson, Giacomo and Viebahn, Konrad and Esslinger, Tilman},
  journal={Nature},
  volume={652},
  number={8110},
  pages={609--614},
  year={2026},
  publisher={Nature Publishing Group}
}

@article{sheng2018high,
  title={High-fidelity single-qubit gates on neutral atoms in a two-dimensional magic-intensity optical dipole trap array},
  author={Sheng, Cheng and He, Xiaodong and Xu, Peng and Guo, Ruijun and Wang, Kunpeng and Xiong, Zongyuan and Liu, Min and Wang, Jin and Zhan, Mingsheng},
  journal={Physical review letters},
  volume={121},
  number={24},
  pages={240501},
  year={2018},
  publisher={APS}
}

@article{wang2025ultrafast,
  title={Ultrafast high-fidelity state readout of single neutral atom},
  author={Wang, Jian and Huang, Dong-Yu and Zhou, Xiao-Long and Shen, Ze-Min and He, Si-Jian and Huang, Qi-Yang and Liu, Yi-Jia and Li, Chuan-Feng and Guo, Guang-Can},
  journal={Physical review letters},
  volume={134},
  number={24},
  pages={240802},
  year={2025},
  publisher={APS}
}

@article{brown2019leakage,
  title={Leakage mitigation for quantum error correction using a mixed qubit scheme},
  author={Brown, Natalie C and Brown, Kenneth R},
  journal={Physical Review A},
  volume={100},
  number={3},
  pages={032325},
  year={2019},
  publisher={APS} 
}

@article{mehta2025bias,
  title={Bias-preserving and error-detectable entangling operations in a superconducting dual-rail system},
  author={Mehta, Nitish and Teoh, James D and Noh, Taewan and Agrawal, Ankur and Anderson, Amos and Birdsall, Beau and Brahmbhatt, Avadh and Byrd, Winfred and Cacioppo, Marc and Cabrera, Anthony and others},
  journal={arXiv preprint arXiv:2503.10935},
  year={2025}
}

@article{ralph2015photon,
  title={Photon sorting, efficient Bell measurements, and a deterministic controlled-Z gate using a passive two-level nonlinearity},
  author={Ralph, TC and S{\"o}llner, I and Mahmoodian, S and White, AG and Lodahl, P},
  journal={Physical Review Letters},
  volume={114},
  number={17},
  pages={173603},
  year={2015},
  publisher={APS}
}

@article{PhysRevX.12.021027,
  title = {Ytterbium Nuclear-Spin Qubits in an Optical Tweezer Array},
  author = {Jenkins, Alec and Lis, Joanna W. and Senoo, Aruku and McGrew, William F. and Kaufman, Adam M.},
  journal = {Phys. Rev. X},
  volume = {12},
  issue = {2},
  pages = {021027},
  numpages = {21},
  year = {2022},
  month = {May},
  publisher = {American Physical Society},
  doi = {10.1103/PhysRevX.12.021027},
  url = {https://link.aps.org/doi/10.1103/PhysRevX.12.021027}
}

@article{PRXQuantum.6.020334,
  title = {High-Fidelity Universal Gates in the ${}^{171}$$\mathrm{Yb}$ Ground-State Nuclear-Spin Qubit},
  author = {Muniz, J. A. and Stone, M. and Stack, D. T. and Jaffe, M. and Kindem, J. M. and Wadleigh, L. and Zalys-Geller, E. and Zhang, X. and Chen, C.-A. and Norcia, M. A. and Epstein, J. and Halperin, E. and Hummel, F. and Wilkason, T. and Li, M. and Barnes, K. and Battaglino, P. and Bohdanowicz, T. C. and Booth, G. and Brown, A. and Brown, M. O. and Cairncross, W. B. and Cassella, K. and Coxe, R. and Crow, D. and Feldkamp, M. and Griger, C. and Heinz, A. and Jones, A. M. W. and Kim, H. and King, J. and Kotru, K. and Lauigan, J. and Marjanovic, J. and Megidish, E. and Meredith, M. and McDonald, M. and Morshead, R. and Narayanaswami, S. and Nishiguchi, C. and Paule, T. and Pawlak, K. A. and Pudenz, K. L. and P\'erez, D. Rodr\'{\i}guez and Ryou, A. and Simon, J. and Smull, A. and Urbanek, M. and van de Veerdonk, R. J. M. and Vendeiro, Z. and Wu, T.-Y. and Xie, X. and Bloom, B. J.},
  journal = {PRX Quantum},
  volume = {6},
  issue = {2},
  pages = {020334},
  numpages = {18},
  year = {2025},
  month = {May},
  publisher = {American Physical Society},
  doi = {10.1103/PRXQuantum.6.020334},
  url = {https://link.aps.org/doi/10.1103/PRXQuantum.6.020334}
}

@misc{senoo2025,
      title={High-fidelity entanglement and coherent multi-qubit mapping in an atom array}, 
      author={Aruku Senoo and Alexander Baumgärtner and Joanna W. Lis and Gaurav M. Vaidya and Zhongda Zeng and Giuliano Giudici and Hannes Pichler and Adam M. Kaufman},
      year={2025},
      eprint={2506.13632},
      archivePrefix={arXiv},
      primaryClass={quant-ph},
      url={https://arxiv.org/abs/2506.13632}, 
}

@article{Ma2023,
    title={High-fidelity gates and mid-circuit erasure conversion in an atomic qubit}, 
    volume={622}, 
    DOI={10.1038/s41586-023-06438-1}, 
    number={7982}, 
    journal={Nature}, 
    author={Ma, Shuo and Liu, Genyue and Peng, Pai and Zhang, Bichen and Jandura, Sven and Claes, Jahan and Burgers, Alex P. and Pupillo, Guido and Puri, Shruti and Thompson, Jeff D.}, 
    year={2023}, 
    month={Oct}, 
    pages={279–284}
}

@article{Zhang_2026,
   title={Logical qubits with erasure conversion using metastable neutral atoms},
   author={Zhang, Bichen and Liu, Genyue and Bornet, Guillaume and Horvath, Sebastian P. and Peng, Pai and Ma, Shuo and Huang, Shilin and Puri, Shruti and Thompson, Jeff D.},
   volume={22},
   ISSN={1745-2481},
   url={http://dx.doi.org/10.1038/s41567-026-03309-0},
   DOI={10.1038/s41567-026-03309-0},
   number={6},
   journal={Nature Physics},
   publisher={Springer Science and Business Media LLC},
   year={2026},
   month={June},
   pages={910–916} }

@article{PhysRevLett.85.2208,
  title = {Fast Quantum Gates for Neutral Atoms},
  author = {Jaksch, D. and Cirac, J. I. and Zoller, P. and Rolston, S. L. and C\^ot\'e, R. and Lukin, M. D.},
  journal = {Phys. Rev. Lett.},
  volume = {85},
  issue = {10},
  pages = {2208--2211},
  numpages = {0},
  year = {2000},
  month = {Sep},
  publisher = {American Physical Society},
  doi = {10.1103/PhysRevLett.85.2208},
  url = {https://link.aps.org/doi/10.1103/PhysRevLett.85.2208}
}

@article{morgado2021quantum,
  title={Quantum simulation and computing with Rydberg-interacting qubits},
  journal={AVS Quantum Science},
  volume={3},
  number={2},
  year={2021},
  doi={doi.org/10.1116/5.0036562},
  publisher={AIP Publishing}
}

@misc{jandura2026,
      title={Surface Code Stabilizer Measurements for Rydberg Atoms}, 
      author={Sven Jandura and Laura Pecorari and Guido Pupillo},
      year={2026},
      eprint={2405.16621},
      archivePrefix={arXiv},
      primaryClass={quant-ph},
      url={https://arxiv.org/abs/2405.16621}, 
}

@unpublished{fumiyoshi_forthcomin2026,
  title     = {Performance Evaluation of Quantum Error Correction Schemes in Dual Ytterbium Systems},
  author    = {Kobayashi, Fumiyoshi and Kusano, Toshi and Fazio, Nicholas and Nakamura, Yuma},
  note      = {in preparation},
}

@article{scholl2023erasure,
  author    = {Scholl, Pascal and Shaw, Adam L. and Tsai, Richard Bing-Shiun and Finkelstein, Ran and Choi, Joonhee and Endres, Manuel},
  title     = {Erasure conversion in a high-fidelity Rydberg quantum simulator},
  journal   = {Nature},
  volume    = {622},
  number    = {7982},
  pages     = {273--278},
  year      = {2023},
  doi       = {10.1038/s41586-023-06516-4},
  url       = {https://doi.org/10.1038/s41586-023-06516-4}
}

@article{debroy2025luci,
  title={LUCI in the surface code with dropouts},
  author={Debroy, Dripto M and McEwen, Matt and Gidney, Craig and Shutty, Noah and Zalcman, Adam},
  journal={Quantum},
  volume={9},
  pages={1936},
  year={2025},
  publisher={Verein zur F{\"o}rderung des Open Access Publizierens in den Quantenwissenschaften}
}

@article{auger2017fault,
  title={Fault-tolerance thresholds for the surface code with fabrication errors},
  author={Auger, James M and Anwar, Hussain and Gimeno-Segovia, Mercedes and Stace, Thomas M and Browne, Dan E},
  journal={Physical Review A},
  volume={96},
  number={4},
  pages={042316},
  year={2017},
  publisher={APS}
}

@article{poulin2005stabilizer,
  title={Stabilizer formalism for operator quantum error correction},
  author={Poulin, David},
  journal={Physical review letters},
  volume={95},
  number={23},
  pages={230504},
  year={2005},
  publisher={APS}
}

@article{PhysRevA.108.03240,
  title = {High-fidelity low-loss state detection of alkali-metal atoms in optical tweezer traps},
  author = {Chow, Matthew N. H. and Little, Bethany J. and Jau, Yuan-Yu},
  journal = {Phys. Rev. A},
  volume = {108},
  issue = {3},
  pages = {032407},
  numpages = {13},
  year = {2023},
  month = {Sep},
  publisher = {American Physical Society},
  doi = {10.1103/PhysRevA.108.032407},
  url = {https://link.aps.org/doi/10.1103/PhysRevA.108.032407}
}

@article{PRXQuantum.4.030337,
  title = {Repetitive Readout and Real-Time Control of Nuclear Spin Qubits in ${}^{171}\mathrm{Yb}$ Atoms},
  author = {Huie, William and Li, Lintao and Chen, Neville and Hu, Xiye and Jia, Zhubing and Sun, Won Kyu Calvin and Covey, Jacob P.},
  journal = {PRX Quantum},
  volume = {4},
  issue = {3},
  pages = {030337},
  numpages = {28},
  year = {2023},
  month = {Sep},
  publisher = {American Physical Society},
  doi = {10.1103/PRXQuantum.4.030337},
  url = {https://link.aps.org/doi/10.1103/PRXQuantum.4.030337}
}

@article{lawrence2004mutually,
  title={Mutually unbiased bases and trinary operator sets for N qutrits},
  author={Lawrence, Jay},
  journal={Physical Review A—Atomic, Molecular, and Optical Physics},
  volume={70},
  number={1},
  pages={012302},
  year={2004},
  publisher={APS}
}

@article{lemke1971set,
  title={Set covering by single-branch enumeration with linear-programming subproblems},
  author={Lemke, Carlton E and Salkin, Harvey M and Spielberg, Kurt},
  journal={Operations Research},
  volume={19},
  number={4},
  pages={998--1022},
  year={1971},
  publisher={INFORMS}
}

@misc{hirai2026hookssurfacecodedistancepreserving,
      title={No More Hooks in the Surface Code: Distance-Preserving Syndrome Extraction for Arbitrary Layouts at Minimum Depth}, 
      author={Yuga Hirai and Shota Ikari and Yosuke Ueno and Yasunari Suzuki},
      year={2026},
      eprint={2603.01628},
      archivePrefix={arXiv},
      primaryClass={quant-ph},
      url={https://arxiv.org/abs/2603.01628}, 
}

@article{7lm4-3bnh,
  title = {Dense packing of the surface code: Code deformation procedures and hook-error-avoiding gate scheduling},
  author = {Fujiu, Kohei and Nagayama, Shota and Nishio, Shin and Kawaguchi, Hideaki and Satoh, Takahiko},
  journal = {Phys. Rev. A},
  volume = {113},
  issue = {4},
  pages = {042412},
  numpages = {14},
  year = {2026},
  month = {Apr},
  publisher = {American Physical Society},
  doi = {10.1103/7lm4-3bnh},
  url = {https://link.aps.org/doi/10.1103/7lm4-3bnh}
}

@article{tomita2014low,
  title={Low-distance surface codes under realistic quantum noise},
  author={Tomita, Yu and Svore, Krysta M},
  journal={Physical Review A},
  volume={90},
  number={6},
  pages={062320},
  year={2014},
  publisher={APS}
}

@inproceedings{vittal2023eraser,
  title={ERASER: Towards adaptive leakage suppression for fault-tolerant quantum computing},
  author={Vittal, Suhas and Das, Poulami and Qureshi, Moinuddin},
  booktitle={Proceedings of the 56th Annual IEEE/ACM International Symposium on Microarchitecture},
  pages={509--525},
  year={2023}
}

@inproceedings{mude2025accurate,
  title={Accurate Leakage Speculation for Quantum Error Correction},
  author={Mude, Chaithanya Naik and Tannu, Swamit},
  booktitle={Proceedings of the 58th IEEE/ACM International Symposium on Microarchitecture},
  pages={579--594},
  year={2025}
}

@article{varbanov2020leakage,
  title={Leakage detection for a transmon-based surface code},
  author={Varbanov, Boris Mihailov and Battistel, Francesco and Tarasinski, Brian Michael and Ostroukh, Viacheslav Petrovych and O’Brien, Thomas Eugene and DiCarlo, Leonardo and Terhal, Barbara Maria},
  journal={npj Quantum Information},
  volume={6},
  number={1},
  pages={102},
  year={2020},
  publisher={Nature Publishing Group UK London}
}

@article{wang2026ai,
  title={AI-Enabled Decoding of Qubit Loss for Quantum Error-Correcting Codes},
  author={Wang, Yuqing and Nie, Xiaotian and Dai, Jiale and Ni, Zhongyi and Zhang, Tao and Zhai, Hui and Chen, Linghui},
  journal={arXiv preprint arXiv:2604.14269},
  year={2026}
}

@article{chen2016measuring,
  title={Measuring and suppressing quantum state leakage in a superconducting qubit},
  author={Chen, Zijun and Kelly, Julian and Quintana, Chris and Barends, Rami and Campbell, Brooks and Chen, Yu and Chiaro, Ben and Dunsworth, Andrew and Fowler, Austin G and Lucero, Erik and others},
  journal={Physical review letters},
  volume={116},
  number={2},
  pages={020501},
  year={2016},
  publisher={APS}
}
